\documentclass[runningheads,envcountsame]{llncs}
\usepackage{graphicx}
\usepackage{amsmath, amssymb}
\usepackage[dvipsnames,svgnames, table]{xcolor}
\usepackage{enumerate}
\usepackage[hidelinks]{hyperref}
\usepackage[disable]{todonotes}
\usepackage{algorithm, algorithmicx, algpseudocode,setspace}
\usepackage[multiple]{footmisc}
\usepackage{thmtools}
\usepackage{thm-restate}
\usepackage{array}

\usepackage[notion, quotation, electronic]{knowledge}

\DeclareMathOperator*{\argmax}{argmax}
\DeclareMathOperator*{\argmin}{argmin}

\renewcommand{\phi}{\varphi}

\renewcommand{\S}{\mathcal{S}}

\newcommand{\NP}{\mathrm{NP}}
\newcommand{\coNP}{\mathrm{coNP}}

\newcommand{\Z}{\mathbb Z}
\newcommand{\R}{\mathbb R}
\newcommand{\N}{\mathbb N}
\newrobustcmd{\Zinfty}{\kl[\Zinfty]{\mathbb Z^{\infty}}}
\knowledge{\Zinfty}{notion}
\newcommand{\Rinfty}{\mathbb R^{\infty}}
\newrobustcmd{\Ninfty}{\kl[\Ninfty]{\mathbb N^{\infty}}}
\knowledge{\Ninfty}{notion}
\newcommand{\re}{\rightarrow}
\newcommand{\rew}[1]{\xrightarrow{#1}}

\newrobustcmd{\VMin}{\kl[\VMin]{V_{\mathrm{Min}}}}
\knowledge{\VMin}{notion}
\newrobustcmd{\VMax}{\kl[\VMax]{V_{\mathrm{Max}}}}
\knowledge{\VMax}{notion}
\newcommand{\val}{\mathrm{val}}

\newrobustcmd{\AttrInf}{\kl[\AttrInf]{\mathrm{Attr}_\infty}}
\knowledge{\AttrInf}{notion}
\newrobustcmd{\Attr}{\kl[\Attr]{\mathrm{Attr}}}
\knowledge{\Attr}{notion}

\newrobustcmd{\MP}{\kl[\MP]{\mathrm{MP}}}
\knowledge{\MP}{notion}
\newrobustcmd{\En}{\kl[\En]{\mathrm{En}}}
\knowledge{\En}{notion}
\newrobustcmd{\Enp}{\kl[\Enp]{\mathrm{En}^+}}
\knowledge{\Enp}{notion}
\newrobustcmd{\Enm}{\kl[\Enm]{\mathrm{En}^-}}
\knowledge{\Enm}{notion}
\newrobustcmd{\First}{\kl[\First]{\mathrm{First}^+}}
\knowledge{\First}{notion}

\newrobustcmd{\valG}{\kl[\valG]{\val_{\game}}}
\newrobustcmd{\FirstG}{\kl[\FirstG]{\mathrm{First}^+_{\game}}}
\newrobustcmd{\EnG}{\kl[\EnG]{\mathrm{En}_{\game}}}
\newrobustcmd{\EnGame}[1]{\kl[\EnGame]{\mathrm{En}_{#1}}}
\newrobustcmd{\EnpG}{\kl[\EnpG]{\mathrm{En}^+_{\game}}}
\knowledge{\valG}[\EnG|\EnpG|\FirstG|\EnGame]{notion}

\newrobustcmd{\PsiFirst}{\kl[\PsiFirst]{\Psi_{\mathrm{First}^+}}}
\knowledge{\PsiFirst}{notion}
\newrobustcmd{\PsiEnp}{\kl[\PsiEnp]{\Psi_{\mathrm{En}^+}}}
\knowledge{\PsiEnp}{notion}
\newrobustcmd{\PsiVal}{\kl[\PsiVal]{\Psi_{\mathrm{val}}}}
\knowledge{\PsiVal}{notion}

\newrobustcmd{\const}{\kl[\const]{\models}}
\knowledge{\const}{notion}

\newcommand{\game}{\mathcal G}
\newrobustcmd{\summ}{\kl[\summ]{\mathrm{sum}}}
\knowledge{\summ}{notion}

\newrobustcmd{\Gphi}[1]{\kl[\Gphi]{\game_{#1}}}
\knowledge{\Gphi}{notion}

\newcommand{\Min}{\mathrm{Min}}
\newcommand{\Max}{\mathrm{Max}}

\newrobustcmd{\negG}[1]{\kl[\negG]{N_{#1}}}
\knowledge{\negG}{notion}

\renewcommand{\neg}{\text{neg}}

\newcommand{\compl}[1]{{#1}^{\mathsf{c}}}

\newrobustcmd{\escF}{\kl[\escF]{\mathsf{esc}_F}}
\newrobustcmd{\escFgame}[1]{\kl[\escFgame]{\mathsf{esc}_F^{#1}}}
\knowledge{\escF}[\escFj|\escFgame]{notion}
\newrobustcmd{\escFj}{\kl[\escF]{\mathsf{esc}_{F_j}}}

\newrobustcmd{\finEsc}{\kl[\finEsc]{F^{<\infty}}}
\knowledge{\finEsc}{notion}
\newrobustcmd{\finEscMax}{\kl[\finEscMax]{F^{<\infty}_{\Max}}}
\knowledge{\finEscMax}{notion}
\newrobustcmd{\finEscMin}{\kl[\finEscMin]{F^{<\infty}_{\Min}}}
\knowledge{\finEscMin}{notion}

\newcommand{\dual}[1]{\overline{#1}}

\newrobustcmd{\wplus}{\kl[\wplus]{w_+}}
\knowledge{\wplus}[\wminus]{notion}
\newrobustcmd{\wminus}{\kl[\wminus]{w_-}}

\newrobustcmd{\PsiGKK}{\kl[\PsiGKK]{\Psi_{\mathsf{GKK}}}}
\knowledge{\PsiGKK}{notion}
\newrobustcmd{\PsiDPPI}{\kl[\PsiDPPI]{\Psi_{\mathsf{DPPI}}}}
\knowledge{\PsiDPPI}{notion}
\newrobustcmd{\PsiQD}{\kl[\PsiQD]{\Psi_{\mathsf{QD}}}}
\knowledge{\PsiQD}{notion}
\newrobustcmd{\PsiOSI}{\kl[\PsiOSI]{\Psi_{\mathsf{OSI}}}}
\knowledge{\PsiOSI}{notion}

\newcommand{\GKK}{\mathsf{GKK}}

\newcommand{\DPPI}{\mathsf{DPPI}}

\newrobustcmd{\leqMon}{\mathrel{\kl[\leqMon]{\leq_{\mathrm{mon}}}}}
\knowledge{\leqMon}{notion}

\definecolor{Blue Marine}{HTML}{022687}
\definecolor{Very Dark Blue}{HTML}{0c1b44}
\definecolor{Dark Green}{HTML}{005715}
\definecolor{Very Dark Green}{HTML}{1e400c}
\definecolor{Dark Magenta}{HTML}{005715}

\definecolor{Dark Ruby Red}{HTML}{580507}
\definecolor{Dark Blue Sapphire}{HTML}{053641}
\definecolor{Dark Gamboge}{HTML}{be7c00}

\definecolor{Blue Sapphire}{HTML}{005f73} 
\definecolor{Blue Matt}{HTML}{094c7b} 
\definecolor{Blue Dark}{HTML}{1B4C6E} 
\definecolor{Golden}{HTML}{E27400}
\definecolor{Brown}{HTML}{ae740e}
\definecolor{Gamboge}{HTML}{ee9b00}
\definecolor{Ruby Red}{HTML}{9b2226}

\IfKnowledgePaperModeTF{
}{
	\knowledgestyle{intro notion}{color={Dark Ruby Red}, emphasize}
	\knowledgestyle{notion}{color={Dark Blue Sapphire}}
	\hypersetup{
		colorlinks=true,
		breaklinks=true,
		linkcolor={}, % Links to sections, pages, etc.
		citecolor={}, % Links to bibliography
		filecolor={Blue Marine}, % Links to local file
		urlcolor={Blue Marine},
	}
}

\knowledge{notion}
  | Simple value iteration
  | SVI

\knowledge{notion}
  | Optimal strategy improvent
  | OSI

\knowledge{notion}
  | alternating versions
  | alternating version
  | PPI-alt
  | DPPI-alt
  | alternating fast value iterations

\knowledge{notion}
  | game

\knowledge{notion}
  | valuation

\knowledge{notion}
  | strategy

\knowledge{notion}
  | value

\knowledge{notion}
  | potential
  | potentials

\knowledge{notion}
  | $\phi $-modified game

\knowledge{notion}
  | potential reduction
  | potential reductions

\knowledge{notion}
  | potential assigner
  | potential assigners

\knowledge{notion}
| $\Psi $-FVI
| $\Psi $-Fast value iteration algorithm
| $\PsiGKK $-fast value iteration
| $\Psi '$-FVI algorithm
| Fast value iteration
| fast value iteration  

\knowledge{notion}
  | Soundness
  | sound
  | soundness

  \knowledge{notion}
  | $\S $-sound
  | $\S _{\mathrm {Trap}}$-sound

\knowledge{notion}
  | Completeness
  | complete
  | completeness

\knowledge{notion}
  | $\infty $-attraction
  | $\infty $-attracting

\knowledge{notion}
  | subgames

\knowledge{notion}
  | $\Min $-trap

\knowledge{notion}
  | DPPI

\knowledge{notion}
  | simple games
  | simple game
  | simple

\knowledge{notion}
  | Positive path iteration algorithm (PPI)
  | PPI
  | Positive path iteration algorithm

\knowledge{notion}
  | escape value

\knowledge{notion}
  | fast value iteration algorithm

\knowledge{notion}
  | Dynamic positive path iteration (DPPI)

\knowledge{notion}
  | QDPM

\knowledge{notion}
  | GKK
  | GKK-DKZ
  | GKK algorithm

\usepackage{marginnote}

\makeatletter
\patchcmd{\@addmarginpar}{\ifodd\c@page}{\ifodd\c@page\@tempcnta\m@ne}{}{}
\makeatother
\reversemarginpar

\let\doendproof\endproof
\renewcommand\endproof{~\hfill$\qed$\doendproof}
\usepackage{lineno}
\newif \iffull
\fulltrue

\begin{document}
\title{Fast value iteration: A uniform approach to efficient algorithms for energy games}
\titlerunning{Fast value iterations for energy games}
% If the paper title is too long for the running head, you can set
% an abbreviated paper title here
%
\author{\inst{}}
\institute{}
\author{Michaël Cadilhac \inst{1} \and Antonio Casares\inst{2}%\orcidID{0000-0002-6539-2020} 
	\and
Pierre Ohlmann\inst{3}} %\orcidID{0000-0002-4685-5253}}
\authorrunning{M. Cadilhac, A. Casares and P. Ohlmann}
% First names are abbreviated in the running head.
% If there are more than two authors, 'et al.' is used.
%
\institute{DePaul University, Chicago, IL, USA \and University of Warsaw, Poland \and CNRS, Université Aix-Marseille, LIS, France}
\maketitle              % typeset the header of the contribution

\begin{abstract}
	We study algorithms for solving parity, mean-payoff and energy games.
	We propose a systematic framework, which we call Fast value iteration, for describing, comparing, and proving correctness of such algorithms.
	The approach is based on potential reductions, as introduced by Gurvich, Karzanov and Khachiyan (1988).
	This framework allows us to provide simple presentations and correctness proofs of known algorithms, unifying the Optimal strategy improvement algorithm by Schewe (2008) and the quasi dominions approach by Benerecetti et al. (2020), amongst others.
	The new approach also leads to novel symmetric versions of these algorithms, highly efficient in practice, but for which we are unable to prove termination.
	We report on empirical evaluation, comparing the different fast value iteration algorithms, and showing that they are competitive even to top parity game solvers.

\end{abstract}

%\vskip1em

%This document contains hyperlinks.
%\AP Each occurrence of a "notion" is linked to its ""definition"".
%On an electronic device, the reader can click on words or symbols (or just hover over them on some PDF readers) to see their definition.

\section{Introduction}\label{sec:introduction}

\paragraph{Mean-payoff and energy games.}
The games under study are infinite duration games where two players, Min and Max, move a token over a finite directed graph with no sink, where the edges of the graph are labelled by payoffs in $\Z$.
When playing a mean-payoff game, the players optimise (minimise or maximise, respectively) the asymptotic average payoff.
In an energy game, they instead optimise the supremum cumulative sum of payoffs within $[0,\infty]$.
These games are positionaly determined~\cite{EM79,BFLMS08}: the two players can play optimally even when restricted to strategies that only depend on the current position of the game.
We refer to Figure~\ref{fig:mean_payoff2} for a complete example.

\begin{figure}[h]
\begin{center}
\includegraphics[width=0.8\linewidth]{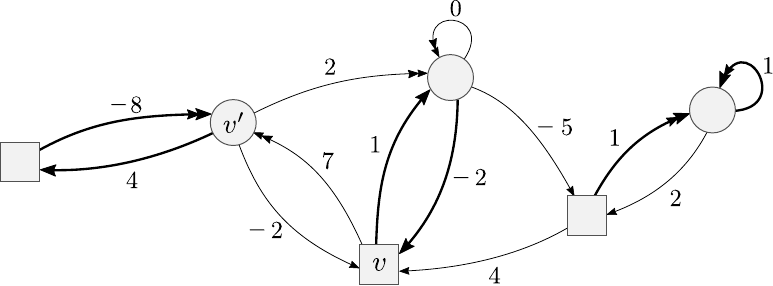}
\end{center}
\caption{Example of a game; circles belong to Min and squares belong to Max.
From left to right, the mean-payoff values are $-2,-2,-\frac 1 2, - \frac 1 2,1$ and $1$, and positional strategies for mean-payoff values are identified in bold.
Energy values are $0,2,9,0,\infty$ and $\infty$, and with optimal strategies given by the double-headed arrows.}
\label{fig:mean_payoff2}
\end{figure}

In this paper, we are interested in the problem of computing energy values of the vertices in a given game which we call \emph{solving} the energy game.
It easily follows from positional determinacy that the energy value of a vertex is finite if and only if its mean-payoff value is non-positive~\cite{BCDGR11}.
Therefore solving an energy game also solves the so called threshold problem for the associated mean-payoff game.
As it turns out, all state-of-the-art algorithms~\cite{BDM24,BV05,BCDGR11,DKZ19,Ohlmann22,Schewe08} for the mean-payoff threshold problem actually solve the energy game.
%The best algorithms for the more general problems of computing the exact values or synthesising optimal strategies in the mean-payoff game also rely on solving many auxiliary energy games~\cite{CR17}.

Mean-payoff values achieved by positional strategies can be computed in polynomial time, and therefore the threshold problem belongs to $\NP \cap \coNP$.
Despite numerous efforts, no polynomial algorithm is known.
Mean-payoff games are known~\cite{Puri95} to generalise parity games~\cite{EJ91,Mostowski:1991} which also belong to $\NP \cap \coNP$ but for which algorithms with quasipolynomial runtime were recently devised~\cite{CJKLS17}.
However, quasipolynomial algorithms for parity games do not generalise to mean-payoff games~\cite{FGO18}.

\paragraph{Algorithmic paradigms.}
There are two well-established paradigms for solving energy games: value iteration (sometimes called ``progress measure'') and strategy improvement.
The standard value iteration for energy games (which we will call ""Simple value iteration"", SVI for short) was introduced by Brim et al.~\cite{BCDGR11}.
While subject to good theoretical (pseudopolynomial) bounds, it is well-known to be prohibitively slow in practice, as its worse-case behaviour is frequently displayed.
On the other hand, strategy improvement algorithms~\cite{BSV04} typically solve practical instances in a constant number of iterations.
Although it offers a useful categorization of older algorithms, the value iteration versus strategy improvement dichotomy fails to accurately describe a new wave of efficient algorithms.

In recent years, multiple hybrid algorithms -- borrowing ideas from both paradigms -- have been put forward.
In~2008, Schewe~\cite{Schewe08} introduced an algorithm called ""Optimal strategy improvent"" (OSI) for solving parity or mean-payoff games. 
As explained by Luttenberger~\cite{Luttenberger08}, Schewe's presentation of OSI is in fact closer to value iteration, but it can also be formally cast as a strategy improvement in a carefully generalised framework allowing for nondeterministic strategies.
In~2019, Dorfman et al.~\cite{DKZ19} presented a value iteration method augmented by a carefully crafted acceleration mechanism (which we call DKZ),  thereby improving on the best theoretical guarantees (this algorithm can be seen as a reformulation of the GKK algorithm~\cite{GKK88}, see also~\cite{Ohlmann22} for further analyses).
Based on the idea of quasi dominions (similar to Fearnley's snares~\cite{Fearnley10a} in a strategy improvement context), Benerecetti et al.~\cite{BDM24} proposed another such acceleration mechanism, obtaining the algorithm QDPM.
% performing extremely well in practice, while preserving state-of-the-art theoretical guarantees.
Some of these algorithms are extremely efficient: a version of OSI is a key component in the LTL-synthesis tool STRIX~\cite{MSL18,LMS20}, which is one of the top competitors in the annual synthesis competition SYNTCOMP~\cite{SyntCompReport22}. On the other hand, QDPM is currently the top-performing mean-payoff game solver~\cite{BDM24} while, remarkably, preserving state-of-the-art theoretical guarantees.

%since its introduction, OSI has been established as a robust parity game solver and it is a key component in the LTL-synthesis tool STRIX~\cite{MSL18,LMS20}, which is one of the top competitors in the main annual synthesis competition SYNTCOMP~\cite{SyntCompReport22}.
%The algorithm was also ported to the GPU by Meyer and Luttenberger~\cite{ML16}.

Although differences in the performances of these algorithms have been observed empirically~\cite{BDM24}, we lack a good understanding of how they compare to each other theoretically, and more generally, of what are the fundamental algorithmic mechanisms that lead to efficient game solvers in practice.

\subsubsection*{Contributions.} Our contributions are as follows.

  \vspace{-2mm}
\paragraph{(1) Fast value iteration framework.}
  We consider "potential reductions", as introduced by Gurvich, Karzanov and Khachiyan~\cite{GKK88}, to design a systematic method for producing algorithms for energy games, which we call the \emph{"fast value iteration" framework}.
  A "potential" is a mapping which assigns a positive weight to each vertex. Such a potential naturally induces a transformation (a "potential reduction") of the game, which preserves the weight of every cycle and thus the values in the mean-payoff game.
  The "fast value iteration" meta-algorithm (Algorithm~\ref{algo-abstrac-potentialReduction}) simply iterates on potential reductions until a fixpoint is reached.
  This meta-algorithm can be instantiated on any given class of "potentials", leading to different algorithms, whose correctness is automatically guaranteed under mild assumptions on the potentials (Theorem~\ref{thm:potential-reduction-algorithm}).
  Interestingly, the framework also provides a \emph{symmetric} meta-algorithm, for which termination is observed in practice, but we have not been able to prove it theoretically.
  %In this framework, each iteration computes a potential, which assigns a positive weight to each vertex, and modifies the game accordingly.
  %Such a transformation preserves the weight of every cycle and thus the values in the mean-payoff game (and thus the winning regions in the energy game) are unaltered.
  %The potential reduction theorem (Theorem~\ref{thm:potential_reduction_theorem}) states that, as long as the potential is positive on some vertex and does not exceed the energy values, the algorithm makes progress.

  The algorithms from the "fast value iteration" framework share some properties that make them convenient for practical applications. The main reason STRIX uses OSI is its support for \emph{modularity}. Since games coming  from LTL-formulas are typically huge, an important feature is to be able to solve them piecewise, avoiding loading the entire game into memory.
  We show that all algorithms within the "fast value iteration" framework are well-suited for this modular approach, which also opens exciting perspectives for parallelised implementations.
  
  \vspace{-2mm}
\paragraph{(2) Unifying and simplifying existing algorithms.} 
	We revisit various algorithms in the light of the above framework. 
	Naturally, the classic SVI~\cite{BCDGR11} is captured (Example~\ref{ex:simpleVI}), as well as the algorithms GKK~\cite{GKK88} and DKZ~\cite{DKZ19} (Section~\ref{subsec:GKK-DKZ}), whose original presentations fit the potential reduction framework.
	
	More interestingly, we also capture algorithms showcasing an excellent performance in practice, defying the common belief that VI algorithms are slow.
	We unify and simplify the algorithms OSI by Schewe~\cite{Schewe08} and the involved QDPM  by Benerecetti et al.~\cite{BDM24}.
	Our presentations are streamlined (see Section~\ref{sec:comparison} for details), leading to immediate correctness proofs.
	It also allows to isolate the core algorithmic idea underlying these two algorithms, which is a natural adaptation of Dijkstra's algorithm to the two-player setting. 
	We call the obtained reinterpretation of OSI and QDPM within the fast value iteration framework, the Positive Path Iteration (PPI).
	  
	The abstraction provided by our approach sets the stage to easily craft new algorithms.
  Showcasing its applicability, we propose a dynamic version of PPI (DPPI), which provably breaks the theoretical barrier set by OSI and QDPM (Lemma~\ref{lem:comparison-potentials}).
%  The conclusion presents numerous possibilities for future work.
  Many possibilities for future work are proposed in the conclusion.
	 %Expressing all algorithms in our unifying framework allows for meaningful comparisons, but also allows to easily craft new algorithms; many possibilities for future work are discussed in the conclusion.
    \vspace{-2mm}
\paragraph{(3) Empirical evaluation.} We compare the implementations of the algorithms described in the fast value iteration framework to OSI and QDPM, as well as to the top parity game solvers.
This evaluation shows: (i) fast value iteration algorithms are highly efficient in practice, and especially robust towards hard instances; (ii) "alternating versions" of the algorithms not only terminate, but are remarkably efficient.

% 
%\paragraph{Outline.}
%Section~\ref{sec:preliminaries} introduces all necessary concepts, recalls the relationship between mean-payoff and energy games and introduces potential reductions.
%Then Section~\ref{sec:fvi} presents the fast value iteration algorithm as well as its symmetric version.
%In  Section~\ref{sec:comparison} we explain how FVI relates to Schewe's algorithm, as well as how we can interpret other existing algorithms in the potential reduction framework.
%Last, we report on our empirical results in Section~\ref{sec:empirical} and then conclude.

\section{Preliminaries}\label{sec:preliminaries}
%\paragraph{Mean-payoff and energy games.}
\AP A ""game"" is a tuple $\game=(G,w,\VMin, \VMax)$, where $G=(V,E)$ is a finite sinkless directed graph, $w:E \to \Z$ is a labelling of its edges by integer weights, and $\VMin,\VMax$ is a partition of $V$.
We set $n=|V|,m=|E|$ and $W=\max\limits_{e\in E} |w(e)|$.
We say that vertices in $\intro*\VMin$ belong to Min and that those in $\intro*\VMax$ belong to Max.
We now fix a game $\game = (G,w,\VMin,\VMax)$.

\AP We simply write $vv'$ for an edge $(v,v')\in E$. A path is a (possibly empty, possibly infinite) sequence of edges $\pi = e_0e_1 \dots$, with $e_{i}=v_{i}v_{i}'$, such that $v_i' = v_{i+1}$.
We write $v_0 \re v_1  \re \dots$ to denote such a path.
%Given a finite or infinite path $\pi=e_0 e_1 \dots$ we let $w(\pi)=w(e_0)w(e_1) \dots$ denote the sequence of weights appearing on $\pi$.
%We sometimes call sequences of weights ``weight profiles''.
\AP The sum of a finite path $\pi$ is the sum of the weights appearing on it, we denote it by $\intro*\summ(\pi)$.
Given a finite or infinite path $\pi=e_0e_1 \dots$ %= v_0 \re v_1 \re \dots$ 
and an integer $k \geq 0$, we let $\pi_{< k} = e_0 e_1 \dots e_{k-1}$, %= v_0 \re v_1 \re \dots \re v_k$, 
and we let $\pi_{\leq k} = \pi_{< k+1}$.
Note that $\pi_{<0}$ is the empty path, and that $\pi_{<k}$ has length $k$. %in general: it belongs to $E^k$.
%We say that $\pi$ starts in $v_0$, and when it is finite and of length $k$ that it ends in $v_{k}$.
By convention, the empty path starts and ends in all vertices.
%A cycle is a finite path which starts and ends in the same vertex.
%We let $\Pi_v^\omega$ denote the set of infinite paths starting in $v$.

\AP A ""valuation"" is a map $\val: \Z^\omega \to \R\cup\{\infty\}$ assigning a potentially infinite value to infinite sequences of weights.
\AP We use $\Rinfty, \intro*\Zinfty$ and $\intro*\Ninfty$ to denote respectively $\R \cup \{\infty\}, \Z \cup \{\infty\}$ and $\N \cup \{\infty\}$.
\AP The four valuations studied in this paper are the mean-payoff, energy, positive-energy, and first-if-positive valuations given by:
\vspace{-1.5mm}
\[\arraycolsep=1.4pt\def\arraystretch{1.6}
\begin{array}{lclclclclcl}
 	\intro*\MP(w) &=& \limsup_k \frac 1 k \sum_{i=0}^{k-1} w_i &\in&  \R,   & \phantom{...} & \intro*\Enp(w)  &=& \sum_{i=0}^{k_\neg - 1} w_i  &\in&  \Ninfty,\\
	\intro*\En(w)  &=& \sup_k \sum_{i=0}^{k-1} w_i  &\in& \Ninfty,  & \phantom{...} & \intro*\First(w)  &=& \max(w_0,0)  &\in&  \N,
\end{array}
\]
where $w=w_0w_1\dots$ is a sequence of weights and $k_\neg = \min \{k \mid w_k <0\} \in \Ninfty$ is the first index of a negative weight.
For technical convenience, we will also consider games in which weights are potentially (positively) infinite.
We extend the definitions of $\En,\Enp$ and $\First$ to words in $(\Zinfty)^\omega$, using the same formula.
Note that for any $w \in (\Zinfty)^\omega$ we have $\Enp \leq \En$.
The four valuations are illustrated on a given sequence of weights in Figure~\ref{fig:valuations}. 

\begin{figure}[ht]
	\begin{center}
		\includegraphics[width=\linewidth]{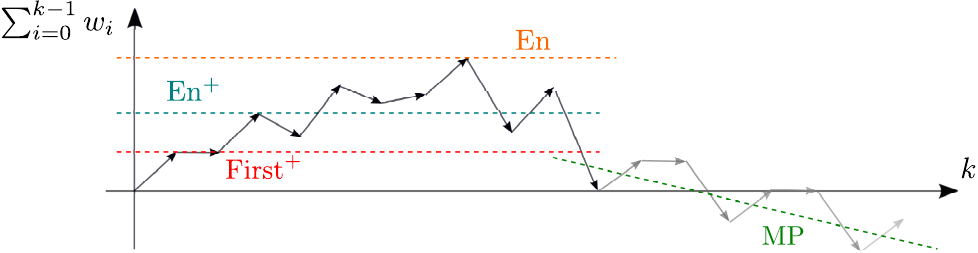}
	\end{center}
	\caption{The three valuations over a given sequence of weights.
		The mean-payoff value is given by the slope of the line, which corresponds to the long-term average.
		In this case, the mean-payoff is $\leq 0$, and both $\En$ and $\Enp$ are finite.}
	\label{fig:valuations}
\end{figure}

\AP A ""strategy"" for Min is a map $\sigma: \VMin \to E$ such that for all $v \in \VMin$, it holds that $\sigma(v)$ is an edge outgoing from $v$.
We say that a (finite or infinite) path $\pi = e_0e_1 \dots$ is consistent with $\sigma$ if whenever $e_i=v_iv_{i+1}$ is defined and $v_i \in \VMin$, it holds that $e_i = \sigma(v_i)$.
\AP We write in this case $\pi \const \sigma$.
\AP Strategies for Max are defined similarly and written $\tau : \VMax \to E$.
%Paths consistent with Max strategies are defined analogously and also denoted by $\pi \const \tau$.
The theorem below states that the three valuations are determined with positional strategies.
It is well known for $\MP$ and $\En$ and easy to prove for $\Enp$. 
We remark that positional determinacy also holds for the two energy valuations $\En$ and $\Enp$ over games where we allow for infinite weights.

\begin{theorem}[\cite{EM79,BFLMS08}]\label{thm:positionality}
	For each $\val \in \{\MP,\En,\Enp\}$, there exist strategies $\sigma_0$ for Min and $\tau_0$ for Max such that for all $v \in V$ we have
	\[
	\sup_{\pi \const \sigma_0} \val(w(\pi)) = \inf_{\sigma} \sup_{\pi \const \sigma} \val(w(\pi)) = \sup_{\tau} \inf_{\pi \const \tau} \val(w(\pi)) = \inf_{\pi \const \tau_0} \val(w(\pi)),
	\]
	where $\sigma, \tau$ and $\pi$ respectively range over strategies for Min, strategies for Max, and infinite paths from $v$.
\end{theorem}

\AP The quantity defined by the equilibrium above is called the ""value"" of $v$ in the $\val$ game, and we denote it by $\intro*\valG(v) \in \Rinfty$;
the strategies $\sigma_0$ and $\tau_0$ are called $\val$-optimal, note that they do not depend on $v$.
\AP The two main algorithmic problems we are interested in are 
(i) computing the value $\En_\game(v)$ of a given vertex $v$ in a game, and 
(ii) decide whether $\MP_\game(v) \leq 0$  (threshold problem).
The following result relates the values in the mean-payoff and energy games; this direct consequence of Theorem~\ref{thm:positionality} was first stated in~\cite{BCDGR11}.

\begin{corollary}[\cite{BCDGR11}]\label{cor:mean_payoffs_and_energies}
	For all $v \in V$ it holds that
	\[
	\MP_\game(v) \leq 0 \iff \EnG(v) < \infty \iff \EnG(v) \leq (n-1)W.
	\]
\end{corollary}

\AP Therefore computing $\En$-values of the games is harder than the mean-payoff threshold problem.
It is easy to deduce $\En$-optimal strategies for Min from the knowledge of the $\En$-values:
we select Min-edges that minimise the sum of the edge's weight and the energy of the destination.
However no knowledge is gained about Max strategies besides the winning region (over which $\En$ values are $\infty$). 
As explained in the introduction, all state-of-the-art algorithms for the threshold problem actually compute $\En$ values.
This shifts our focus from mean-payoff to energy games.

\paragraph{Attractors.}
\AP Given a subset $S \subseteq V$, the attractor $\intro*\Attr_\game^\Max(S)$ to $S$ in $\game$ is defined to be the set of vertices $v$ such that Max can ensure to reach $S$ from $v$.

\paragraph{Simple games.}
A finite path $v_0 \re v_1 \re \dots \re v_k$ is simple if there is no repetition in $v_0,v_1,\dots,v_{k-1}$; note that a cycle may be simple.
\AP A game is ""simple"" if all simple cycles have nonzero sum.
The following result is folklore and states that one may reduce to a simple game at the cost of a linear blow up on $W$.
It holds thanks to the fact that positive mean-payoff values are $\geq 1/n$ (with $n=|V|$), which is a well-known consequence of Theorem~\ref{thm:positionality}.

\begin{lemma}\label{lem:lifting_simplicity}
	Let $\game=(G,w,\VMin,\VMax)$ be an arbitrary game.
	The game $\game'=(G,w',\VMin,\VMax)$, with $w' = (n+1)w-1$, is simple and has the same vertices of positive mean-payoff values as $\game$.
\end{lemma}

\section{Fast value iteration: A meta-algorithm based on potential reductions}\label{sec:potential_reductions}

\subsection{Potential reductions}\label{subsec:potential-reductions}
Fix a game $\game = (G=(V,E),w,\VMin,\VMax)$.
\AP A ""potential"" is a map $\phi:V \to \Ninfty$.
Potentials are partially ordered coordinatewise. We write $\phi = 0$ if $\phi(v)=0$ for all $v\in V$.
Given an edge $vv' \in E$, we define its $\phi$-modified weight to be
\[
w_\phi(vv') = \begin{cases} 
	\infty & \text{ if } \phi(v),\phi(v') \text{ or } w(vv') \text{ is } \infty, \\
	w(vv') + \phi(v') - \phi(v) & \text{ otherwise}.
\end{cases}
\]
\AP The ""$\phi$-modified game"" $\intro*\Gphi{\phi}$ is simply the game $(G,w_\phi,\VMin,\VMax)$; informally, all weights are replaced by the modified weights.
Note that the underlying graph does not change, in particular paths in $\game$ and $\game_\phi$ are the same.
\AP Moving from $\game$ to $\game_\phi$ for a given potential $\phi$ is called a ""potential reduction"".

Weights of cycles are preserved by finite potential reductions, and therefore, as an easy consequence of positionality (Theorem~\ref{thm:positionality}), mean-payoff values are preserved.
Note that any edge outgoing from a vertex $v$ with potential $\phi(v) = \infty$ has weight $\infty$ in the modified game, therefore $v$ has $\En$ and $\Enp$-values $\infty$ in $G_\phi$.
Note also that sequential applications of potential reductions correspond to reducing with respect to the sum of the potentials: $(\game_\phi)_{\phi'} = \game_{\phi+\phi'}$.

Potential reductions were introduced by Gallai~\cite{Gallai58} for studying network-related problems such as shortest-paths problems.
In the context of mean-payoff or energy games, they were introduced in~\cite{GKK88} and later sometimes rediscovered.
%The main result that allows to use potential reductions to solve energy games is presented in Theorem~\ref{thm:potential_reduction_theorem} below. It describes the effect of potential reductions over energy values. We will use the following lemma to prove Theorem~\ref{thm:potential_reduction_theorem}.

\subsection{The fast value iteration meta-algorithm}
\AP A ""potential assigner"" is a function $\Psi$ that assigns a potential $\Psi(\game):V \to \Ninfty$ to each game $\game$. 
A "potential assigner" $\Psi$ induces a \emph{fast value iteration algorithm} (called "$\Psi$-FVI") as follows:
successively apply potential reductions using the potentials given by $\Psi$, until a game $\game'$ is reached with $\Psi(\game')(V) \subseteq \{ 0, \infty\}$.
For an arbitrary potential assigner, this algorithm might not terminate, or provide a final game $\game'$ carrying irrelevant information.
However, we show that under mild hypotheses on $\Psi$, this algorithm terminates, and $\En_{\game'} = \Psi(\game')$, with the vertices with $\En$-value $0$ corresponding to the vertices with finite value in the original game.
Moreover, the exact $\En$-values of the original game can be recovered from the sequence of potentials obtained during the computation.

We formalise this idea in Algorithm~\ref{algo-abstrac-potentialReduction} and Theorem~\ref{thm:potential-reduction-algorithm}. To ensure termination, we need to artificially increase the potential of some vertices to $\infty$ whenever a threshold is reached, and then remove Max's attractor to $\infty$.
This technique is standard in value iteration algorithms, see e.g.~\cite{BCDGR11}.

\algdef{SE}[DOWHILE]{Do}{doWhile}{\algorithmicdo}[1]{\algorithmicwhile\ #1}%
\begin{algorithm}

	\caption{\AP ""$\Psi$-Fast value iteration algorithm"".}
	\label{algo-abstrac-potentialReduction}
	\begin{algorithmic}[1]
		\Statex \textbf{Input:} Game $\game$ with $n$ vertices and maximal weight $W$
		\State $\Phi \gets 0$ \Comment{$\Phi$ carries the cumulative sum of potentials}
		\Do
%		\While{$\phi(v) \neq 0$ and $\game \neq \emptyset$}\label{line:while}
			\State $\phi \gets \Psi(\game)$ %(compute a potential for the new game)
			\State $\Phi \gets \Phi + \phi$ \Comment{Update cumulative sum over $\game$\footnotemark}\label{line:update-cumulative}
			\State $\game \gets \game_\phi$\label{line:update}	%(apply a potential reduction according to $\phi$)
			\State $A \gets \Attr_{\game}^{\Max}(\Phi^{-1}([(n-1)W+1,\infty]))$\label{line:artificial-infty}
			\State Set $\Phi(v) = \infty$ for all $v \in A$	\label{line:remove-attractor}
			\State $\game \gets \game \setminus A$\label{line:end-of-while}
			%(remove $\Max$'s attractor to $\infty$)
			
		\doWhile{$\phi \neq 0$ and $\game \neq \emptyset$}
		%\State Let $\sigma_\Min: \VMin \to E$ be a function taking edges with weight $\leq 0$ in $\game$ whenever such an edge exists.
		%\State Let $\sigma_\Max: \VMax \to E$ be a function taking edges leading to vertices $v$ with $\Phi(v)=\infty$ whenever such an edge exists.
		%\State \Return $(\Phi, \sigma_\Min, \sigma_\Max)$
		\State \Return $\Phi$
	\end{algorithmic}
\end{algorithm}
\footnotetext{By a small abuse of notation, we allow to sum potential with different domains. If $\phi\colon V \to \Ninfty$ and $\phi'\colon V' \to \Ninfty$ with $V'\subseteq V$, then $\phi + \phi'(v) = \phi(v)$ for all $v\notin V'$.  }

Let us isolate two relevant properties of potential assigners:
(1) \AP ""Soundness"": for any game $\game$, $\Psi(\game) \leq \EnG$;
(2) \AP ""Completeness"":  for any $\game$, if $\Psi(\game) = 0$ then $\EnG=0$.
We also say that a potential $\phi$ is "sound" over a given game if condition (1) is met.
We may now state our first main result.

%\vskip1em
%
%\begin{tabular}{rl}
%	""Soundness:""& for any game $\game$, $\Psi(\game) \leq \EnG$.\\
%	""Completeness:"" & for any $\game$, if $\Psi(\game) = 0$ then $\EnG=0$.
%\end{tabular}
%\vskip1em

\begin{restatable}{theorem}{metaAlg}\label{thm:potential-reduction-algorithm}
	Let $\Psi$ be a "sound" and complete potential assigner.
	Then Algorithm~\ref{algo-abstrac-potentialReduction} terminates in at most $n^2W$ iterations, and returns $\Phi = \EnG$.
\end{restatable}

\begin{remark}
	Note that the hypotheses of the theorem are minimal.
	If a potential assigner $\Psi$ is not sound, there is a game $\game$ for which the algorithm returns $\Phi \geq \Psi(\game) > \EnG$.
	If it does not satisfy (ii), there is a game for which the algorithm stops in the first iteration, returning the potential $\Phi = 0 \neq \EnG$.
\end{remark}

\begin{example}[Simple value iteration of Brim et al.~\cite{BCDGR11}]\label{ex:simpleVI}
	Define the potential assigner $\intro*\PsiFirst$ by assigning the potential $\FirstG(v)$, the first-if-positive value, to a vertex $v$.
	This potential is easily computed in linear time as it coincides for each Max (resp. Min) vertex $v$, with the maximal (resp. minimal) value of $\max(w,0)$ where $w$ ranges over outgoing weights.
	Clearly $\PsiFirst \leq \EnG$, since for any sequence of weights $w_0w_1\dots$, it holds that $\First(w_0w_1\dots) \leq \En(w_0w_1\dots)$.
	Finally, if $\PsiFirst=0$, then from any vertex Min can ensure that no positive weight is ever seen, which entails $\EnG=0$. We conclude that $\PsiFirst$ is sound and complete; the fast value iteration algorithm coincides with that of~\cite{BCDGR11}.\footnote{Formally, reducing from complexity $O(n^2mW)$ to $O(nmW)$ requires some additional bookkeeping.}
\end{example}

\begin{example}\label{ex:valuation-yields-potential}
	\AP Any (determined) valuation $\val\colon \Z^\omega \to \Rinfty$  induces a potential assigner $\intro*\PsiVal$, namely, the one that assigns to each game $\game$ the potential given by $\val_\game(v)$. %, or extend the definition by $\val_\game(v) = \infty$ whenever this quantity is not well-defined.}
	If the valuation satisfies $\val\leq \En$ over weight sequences, then $\PsiVal$ is sound.
	Moreover, if $\val(w_0w_1...) >0$ whenever $w_0>0$, then $\PsiVal$ is complete.
	This includes the previous example, 
	%instanciated with the first-if-positive valuation.
	and more interestingly, this includes the valuation $\Enp$, which is the subject of Section~\ref{subsec:PPI}.
\end{example}

Of course, an important requirement over $\Psi$ to make Algorithm~\ref{algo-abstrac-potentialReduction} relevant is that we should be able to compute $\Psi(\game)$ efficiently.
Note that the potential assigner corresponding to the $\En$-values of a game satisfies all the required hypothesis, and makes Algorithm~\ref{algo-abstrac-potentialReduction} terminate in a single iteration. 

\vspace{-1mm}
\paragraph{$\infty$-attraction.} \AP In many occurrences, the algorithm can be simplified by removing lines 6-8 and stopping when a fixpoint is reached (which can be implemented by replacing line 9 with ``{while} $\game_{\Psi(\game)} \neq \game$'').
We say that potential assigners with this property are ""$\infty$-attracting"".
\iffull
We provide easy-to-check sufficient conditions for "$\infty$-attraction"  in Appendix~\ref{app:infAttracting}.
\else\fi

\vspace{-1mm}
\paragraph{Modularity.} %As noted in the introduction (a version of) OSI is the algorithm used by STRIX~\cite{MSL18,LMS20}, the top LTL-synthesis tool in SYNTCOMP~\cite{SyntCompReport22}. However, it is well-known that other parity game solvers are faster over most instances (c.f. Section~\ref{sec:empirical}). Then, why OSI? The main reason is that it allows for modularity (it was ported to the GPU by Meyer and Luttenberger~\cite{ML16}). Games coming from LTL-formulas are typically huge, one key feature of STRIX is the ability to solve games modularly, avoiding loading  the entire games into memory.
%The algorithm was also ported to the GPU by Meyer and Luttenberger~\cite{ML16}.
The framework of "fast value iteration" is specially well suited for a modular approach, allowing to solve games piecewise, as we show next.
%We note that, since $\game$ and $\game_{\phi}$ have the same positive cycles, for all finite $\phi$, we can safely apply any potential reduction to a subgame without modifying the winning region of the players in the global game. 
%To moreover preserve the $\En$-values, or allow for potentials taking the value $\infty$, we need a further requirement.
A subgame is a pair $(\game',\game)$, with $\game'\subseteq \game$. Let $\S$ be a class of ""subgames"". We say that a "potential assigner" $\Psi$ is ""$\S$-sound"" if for all subgames $(\game',\game)\in \S$ it holds that $\Psi(\game') \leq \EnG|_{\game'}$, that is, the potential is "sound" when applied only to this part of the game. 
(We note that for instance, any "sound" potential is "$\S_{\mathrm{Trap}}$-sound" for the class of subgames $(\game',\game)$ such that $\game'$ is a "$\Min$-trap".) %(as $\En_{\game'}\leq \EnG|_{\game'}$ in that case).
Therefore, if $\Psi$ is "$\S$-sound", we can solve subgames in $\S$ partially, and apply the corresponding "potential reduction" in the whole game, progressing towards a computation of the $\En$-values.\\[-1mm]

The rest of the section is devoted to a proof of Theorem~\ref{thm:potential-reduction-algorithm}.
\iffull
Detailed proofs are available in Appendix~\ref{app:metaAlg-correctness}.
\else\fi

\vspace{-1mm}

\paragraph{Termination.} Termination of Algorithm~\ref{algo-abstrac-potentialReduction} is ensured thanks to lines~\ref{line:artificial-infty} and~\ref{line:remove-attractor}: the function $\Phi$ strictly increases in each non-terminating iteration, and it only takes values in $[0,nW]\cup \{\infty\}$, hence the bound $n^2W$.

\vspace{-1mm}

\paragraph{Correctness.}
We now state the key technical theorem enabling our framework.
It describes the effect of "sound" potential reductions over energy values, allowing to combine them. 
From it, we easily derive compositionality of "sound" potentials.

\begin{restatable}[Update of energy values]{theorem}{potentialReduction}\label{thm:potential_reduction_theorem}
	If $\phi$ is "sound" then $\EnG = \phi + \En_{\game_\phi}$.
\end{restatable}

\begin{corollary}[Compositionality]\label{cor:compositionality}
	If $\phi$ is "sound" for $\game$ and $\phi'$ is "sound" for $\game_\phi$ then $\phi + \phi'$ is "sound" for $\game$.
\end{corollary}
\begin{proof}
	As $\phi'$ is "sound" for $\game_\phi$, we have that $\phi' \leq \En_{\game_\phi}$.
	Adding $\phi$ on both sides, we get
	$\phi + \phi' \leq \phi + \En_{\game_\phi}$.
	By Theorem~\ref{thm:potential_reduction_theorem}, the right hand-side is equal to $\EnG$, as desired.
\end{proof}

We are now ready to prove Theorem~\ref{thm:potential-reduction-algorithm}.
\iffull
(The formal proof requires a bit more work regarding vertices sent to $\infty$, see Appendix~\ref{app:metaAlg-correctness}.)\else\fi

\begin{proof}[Informal proof for Theorem~\ref{thm:potential-reduction-algorithm}]
	Let $\game_i$, $\phi_i=\Psi(\game_i)$ and $\Phi_i=\phi_1 + \dots + \phi_{i-1}$ denote the game, potential and cumulative sum at the $i$-th iteration of the algorithm.
	Since $\Psi$ is sound, $\phi_i$ is "sound" for $\game_i$ for all $i$.
	Thus it follows from an easy induction and compositionality that for all $i$, $\Phi_i$ is "sound" for $\game$.
	In particular, for the maximal $i$, Theorem~\ref{thm:potential_reduction_theorem} gives $\Phi_i + \En_{\game_i} = \EnG$, but moreover since $\phi_i=0$ we get by completeness that $\En_{\game_i}=0$ which concludes.
\end{proof}

\subsection{Asymmetry and alternating fast value iteration}\label{subsec:alternating}

"Fast value iteration" is based on successive underapproximations of the energy valuation $\En$, which is inherently asymmetric.
However, the initial problem (solving mean-payoff games) is itself symmetric, which calls for the design of more symmetrical solutions, a recurring theme in the literature~\cite{JMOT20,JMT22,STV15,TOJ22}.

\paragraph{Dual algorithm computing $\Max$-values.}
Let $\intro*\dual{\game}$ be the game obtained by swapping $\VMin$ and $\VMax$ and relabelling the weights by $\dual{w} = -w$.
\AP The two games are essentially equivalent, for instance the mean-payoff values in $\dual{\game}$ and $\game$ are opposite.
However asymmetric algorithms such as value iterations behave differently over each game; this is useful for instance if one wants to compute Max strategies in $\game$, which are output by running value iterations in the dual.
But this still does not provide a symmetric solution.

\paragraph{Alternating fast value iteration.}
\AP We now consider ""alternating versions"" of the algorithm, by working with potentials in $\phi : V \to \Z \cup \{\pm \infty\}$.
The algorithm applies potential reductions corresponding to $\Psi$ and its dualized version $\dual{\Psi}$ \emph{on the same game} in an alternating fashion%\footnote{This requires to take care of some technicalities concerning $\pm \infty$.}
, until all vertices are sent to $+\infty$ or $-\infty$.
Naturally, when a vertex is set to $+\infty$ or $-\infty$, the adequate attractor is computed and removed from the game.

Assuming the "potential assigner" $\Psi$ is "sound", since "sound" "potential reductions" do not alter winning regions, the algorithm is correct and Min's winning region is the preimage of $-\infty$ by the final potential.
Termination, however, is not easily guaranteed.
Interestingly, we observe experimentally that, for some potential assigners, this alternating algorithm always terminates, and it is even remarkably fast (see Section~\ref{sec:empirical}).
We leave as an interesting open problem to determine for which potential assigners (if any) this algorithm terminates (see conclusion).

%\po{Ce serait intéressant de voir expériementalement si ca termine pour First (aussi pour QD). Dans ce cas on pourrait même dire qu'on observe que ca termine pour tous les potential assigners du papier, et donc poser la question en général.}

%
%
%We now present a symmetric version of the fast value iteration algorithm.
%This requires working with potentials that may take negative values $\phi : V \to \Z$.
%The key idea is that one may use Dijkstra's algorithm to compute, in the very same fashion, the $\Enm$-values of the game, where $\Enm : \Z^\omega \to [-\infty,0]$ is given by
%$
%\Enm(w_0 w_1 \dots) = \sum_{i=0}^{k_{\pos}-1} w_i
%$,
%with $k_\pos = \min \{k \mid w_k >0\}$.

%\section{Fast value iteration}\label{sec:fvi}
%\input{fastValueIteration}

\section{Instances of fast value iteration and theoretical comparisons}\label{sec:comparison}

We have already shown (Example~\ref{ex:simpleVI}) how SVI instantiates in our framework.
In this section, we introduce further "potential assigners" to capture known efficient algorithms for energy games, and prove their "soundness".
This provides a streamlined and unified presentation of (versions of) the algorithms OSI~\cite{Schewe08} and QDPM~\cite{BDM24} (Section~\ref{subsec:PPI}), namely the positive path iteration algorithm (PPI).
%This allows us to (i) obtain immediate correctness proofs, (ii) present a symmetric version of OSI, very efficient in practice, and (iii) isolate similarities between these algorithms.
We also propose a dynamic variant "DPPI", corresponding to a "potential assigner" generating "potentials" with provably larger values.
At the end of the section we also discuss the GKK algorithm~\cite{GKK88}, and then provide formal comparisons between the four algorithms stated in our framework.% whose original presentation can be considered to already be a "potential" reduction algorithm. 

In all cases, we find that the algorithms are easier to explain over "simple games", which we will assume without loss of generality (see Lemma~\ref{lem:lifting_simplicity}); note also that simplicity is preserved by "potential" reductions.

\subsection{The positive path iteration algorithm}\label{subsec:PPI}

We now study the "fast value iteration" algorithm corresponding to the "potential assigner" $\intro*\PsiEnp(\game)=\EnpG$.
\AP We call it the ""Positive path iteration algorithm (PPI)"".
It is immediate to check that the "potential assigner" $\PsiEnp$ is "sound" and complete (see Example~\ref{ex:valuation-yields-potential}), so Theorem~\ref{thm:potential-reduction-algorithm} applies, directly giving correctness of "PPI".
Moreover, we can in this case simplify the algorithm by removing lines~\ref{line:artificial-infty}-\ref{line:remove-attractor} in Algorithm~\ref{algo-abstrac-potentialReduction}, because $\PsiEnp$ is "$\infty$-attracting".
\iffull
We refer to Appendix~\ref{app:infAttracting} for a proof of this fact.
\else\fi
%as explained in Section~\ref{sec:potential_reductions}

\begin{proposition}\label{prop:PPI-is-infAttracting}
	The "potential assigner" $\PsiEnp$ is "$\infty$-attracting".
\end{proposition}

We let $N_\game$ denote the set of vertices from which Min can ensure to immediately see a negative vertex: $v \in \VMax$ (resp. $\VMin$) belongs to $N$ if and only if all outgoing edges (resp. some outgoing edge) have weight $<0$.
Note that computing $\Enp$-values in $\game$ corresponds to solving a variant of the energy game which stops whenever $N$ is reached.
It turns out that this problem is (efficiently) tractable, thanks to two-player game extensions of Dijkstra's algorithm.
In fact, two seemingly distinct algorithms are known, corresponding to OSI~\cite{Schewe08} and QPDM~\cite{BDM24}.
Remarkably, Khachiyan, Gurvich and Zhao~\cite{KGZ06} solved the same problem\footnote{This corresponds to Theorem 1 in~\cite{KGZ06}, case $(i)$ with blocking systems $\mathcal B_2$.} earlier and in a different context (with an algorithm similar to Schewe's).

\paragraph{Two algorithms for computing $\Enp$.} We now describe the two algorithms, respectively extracted from~\cite{Schewe08} and~\cite{BDM24}. We first introduce some notation. 
\AP For a subset $F\subseteq V$, a vertex $v\in F$ and an edge $vv'$, we define $\escF(vv') = w(vv')$ if $v'\notin F$, and $\escF(vv') = \infty$, if $v'\in F$. We define the "escape value" of a vertex as:
\[ \intro*\escF(v) = \begin{cases}
	\min\{\escF(vv') \mid w(vv') \geq 0\}, \text{ if } v\in \VMin,\\
	\max\{\escF(vv') \mid w(vv') \geq 0\}, \text{ if } v\in \VMax.
\end{cases} \]

We will only consider subsets $F \subseteq \compl{\negG{\game}}$, so Max vertices have a non-negative outgoing edge and Min vertices have only non-negative outgoing edges, in particular $\escF(v) \geq 0$.
It can be seen as the minimal weight that $\Min$ can force to see while leaving~$F$ immediately from $v$, or $\infty$ is she cannot force to leave $F$ in one step, assuming Max is constrained to playing non-negative edges.
\AP We further let $\intro*\finEsc$ denote the set of vertices with finite $\escF$, and $\intro*\finEscMax$ and $\intro*\finEscMin$ their intersections with $\VMax$ and $\VMin$.
Last, for $v\in V$, the notation $\phi_v(v) \gets x$ indicates that $\phi_v$ is the "potential" defined by $\phi_v(v) = x$ and $\phi_v(v') = 0$ for $v'\neq v$. 

\vspace{-5mm}
\hspace{-0.65cm}
\begin{minipage}[t]{0.52\textwidth}
	\begin{algorithm}[H]
	\centering
	\caption{Subprocedure in OSI}\label{alg:PPI_OSI}
		\begin{algorithmic}
		\Statex \textbf{Input:} Simple game $\game$
		\State $\Phi \gets 0$
		\State $F\gets V \setminus N_\game$
		\While{$\finEsc \neq \emptyset$}
			\If{$\finEscMax \neq \emptyset$}
				\State let $v \in \finEscMax$
				\State $\Phi(v) \gets \max\limits_{w(vv')\geq0} w(vv')+\Phi(v')$
			\ElsIf{$\finEscMin \neq \emptyset$}
				\State let $v \in \finEscMin$ minimizing \\ $\qquad \qquad \qquad m=\min\limits_{v' \in F} w(vv')+\Phi(v')$
				\State $\Phi(v) \gets m$
			\EndIf
		\EndWhile
		\State $\Phi_v(v) \gets \infty$ for all $v\in F$
		\State \Return $\Phi$
	\end{algorithmic}
	
	\end{algorithm}
	\end{minipage}
	\hfill
	%\hspace{-0.1cm}
	\begin{minipage}[t]{0.48\textwidth}
		%\strut\vspace*{-2.1cm} %%%%% CHANGER ICI À LA FIN POUR ALIGNEMENT VERTICAL
		\begin{algorithm}[H]
			\centering
			\caption{Subprocedure in QDPM}\label{alg:PPI_QD}
		\begin{algorithmic}
			\Statex \textbf{Input:} Simple game $\game$
			\State $\Phi = 0$
			\State $F\gets V \setminus N_\game$
			\While{$\finEsc \neq \emptyset$}
			%\State let $v\in F$ with minimal $\escF(v)$
			\State let $v\in \argmin_F\escF(v)$
			\State $\phi_v(v) \gets \escF(v)$\label{line:pickv}
			\State $F \gets F\setminus \{v\}$
			\State $\game \gets \game_{\phi_v}$
			\State $\Phi = \Phi + \phi_v$
			\EndWhile
			\State $\Phi_v(v) \gets \infty$ for all $v\in F$
			\State \Return $\Phi$
		\end{algorithmic}
		\end{algorithm}
	\end{minipage}

\vskip1em

\begin{restatable}[Adapted from~\cite{KGZ06,Schewe08,BDM24}]{theorem}{correctnessPPI}\label{thm:correctnessPPI}
	Algorithms~\ref{alg:PPI_OSI} and~\ref{alg:PPI_QD} both compute $\EnpG$, and both can be implemented to run in $O(m+n\log n)$ operations.
\end{restatable}

\iffull
A detailed proof is given in Appendix~\ref{app:correctnessPPI}.
In Appendix~\ref{app:comparisons} we include detailed comparisons between the "Positive path iteration algorithm" (PPI), and the algorithms OSI and QDPM.
\else\fi

\subsection{A new fast value iteration algorithm}\label{subsec:new-DPPI}

\AP Drawing inspiration from Algorithms~\ref{alg:PPI_OSI} and~\ref{alg:PPI_QD} above, we introduce another "potential assigner", leading to a "fast value iteration algorithm" which we call ""Dynamic positive path iteration (DPPI)"".
Note that both algorithms above compute the Min attractor to $\negG{\game}$ over non-negative edges, which corresponds exactly to the set of vertices with finite $\Enp$, and obtain the values of $\Enp$ by backtracking.
We will also backtrack over the same attractor, and just as in Algorithm~\ref{alg:PPI_QD}, we make "potential" updates on the fly.
The difference is in the precise way in which we choose the vertices, which enables in our case that some of the "potential" updates may cause new edges to become positive, which will then be taken into account, sometimes leading to a "potential" $>\Enp$.
%To ensure "soundness", we combine the two principles of Algorithms~\ref{alg:QD-potential} and~\ref{alg??}: add a Min vertex to the attractor if possible ($v$ has minimal "escape value"), otherwise add a Max-vertex.
%We note that $\argmax_{v\in \finEscMax}\escF(v)$ is the set of $\Max$-vertices with maximal, finite, "escape value".

\begin{algorithm}
	\caption{Computation of the $\intro*\PsiDPPI$-potential.}
	\label{alg:DPPI}
	\begin{algorithmic}[1]
		\Statex \textbf{Input:} Simple game $\game$
		\State $F\gets V \setminus \negG{\game}$
		\While{$\finEsc \neq \emptyset$}
		\If{there is $v\in (\argmin \escF(v))\cap \VMin$}
		{$\phi_v(v) \gets \escF(v)$}
		\Else{ let $v\in \argmax_{v\in \finEscMax} \escF(v)$ and $\phi_v(v) \gets \escF(v)$}
		\EndIf \label{line:endIf}
		\State $F \gets F\setminus \{v\}$
		\State $\game \gets \game_{\phi_v}$
		\EndWhile
		\State $\phi_v(v) \gets \infty$ for all $v\in F$\label{line:afterWhile-DPPI}
		\State \Return $\sum_{v\in V}\phi_v$
	\end{algorithmic}
\end{algorithm}

%The following is proved in Appendix~\ref{app:correctnessDPPI}. A formal comparison with PPI is made below.

\begin{restatable}{lemma}{correctnessDPPI}
	The "potential assigner" $\PsiDPPI$ is "sound" and "complete" for simple games.
\end{restatable}

We see DPPI as a marginal improvement over PPI, but an improvement nonetheless, showing that the barrier imposed by PPI can be broken, motivating future work.
Figure~\ref{fig:exampleDPPI} shows a game where DPPI performs fewer iterations than PPI, while Lemma~\ref{lem:comparison-potentials} below proves that for any game $\game$, $\Psi_\DPPI(\game) \geq \Psi_{\Enp}(\game)$.

\begin{figure}[h]
	\begin{center}
		\includegraphics[width=0.34\linewidth]{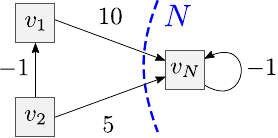}
		\caption{A game $\game$ (all vertices belong to $\Max$) where DPPI performs a single iteration, as $\PsiDPPI(\game)=[v_1 \to 10;v_2\to 9;v_N \to 0]=\En_\game$. In contrast, PPI requires two iterations since $\Enp=[v_1 \to 10;v_2\to 5;v_N \to 0]$.}\label{fig:exampleDPPI}
	\end{center}
	\vspace{-8mm}
\end{figure}

\subsection{The GKK algorithm}\label{subsec:GKK-DKZ}

We include a short discussion about the GKK algorithm; a more detailed modern exposition, including state-of-the-art upper bounds and comparison with the related approach of Dorman et al.~\cite{DKZ19}, was proposed by Ohlmann~\cite{Ohlmann22}.

\AP The ""GKK algorithm"" is the "$\PsiGKK$-fast value iteration" where $\intro*\PsiGKK$ is the "potential assigner" defined as follows.
% associated with the "potential assigner" $\PsiGKK$, defined as follows.

Let $V_{-}$ be the set of vertices from which Min can ensure that a negative edge is seen before the first positive edge. (Note that $V_{-}$ coincides with $(\EnpG)^{-1}(0)$.)
Likewise, let $V_{+}$ denote the set of vertices from which Max can ensure seeing a positive edge before a negative one; and observe that in a simple game, $V_{+}$ is the complement of $V_{-}$.
\AP Consider the maximal value $\intro*\wplus$ such that from any vertex of $V_{+}$ Max can ensure to add up to $\wplus$ before a negative weight is seen (alternatively, $\wplus$ is the smallest nonzero value of $\EnpG$) ; and dually for $\wminus$.
%Then compute the smallest value $\wplus$ of a positive weight that Max can ensure to see before a negative weight when starting from $V_{+}$ (alternatively,
Note that, if from any vertex in $V_{+}$, Max can ensure to remain in $V_{+}$ while seeing positive vertices, then $\wplus=\infty$.
Clearly $\wplus \leq \Enp \leq \En$ over $V_{+}$.
We define $\PsiGKK(\game)(v)$ to be $\min(\wplus,-\wminus)$ if $v \in V_{+}$ and $0$ otherwise.
Soundness follows from the inequality above, and completeness is easy to prove.
\iffull
Moreover, $\PsiGKK$ is "$\infty$-attracting" (a proof of this fact is included in Appendix~\ref{app:infAttracting}).
\else
Moreover, $\PsiGKK$ is "$\infty$-attracting".
\fi

The "potential"  $\PsiGKK$ has a remarkable symmetric property: the assigned potentials are the same over $\game$ and over its dual $\dual{\game}$: $\PsiGKK = \dual{\PsiGKK}$.\footnote{This was first observed by Ohlmann~\cite{Ohlmann22} leading to an improved upper bound.}
In particular, the algorithm and its "alternating version" coincide.

\subsection{Comparing fast value iteration algorithms}
We now propose formal comparisons between the above "potential assigners".
Intuitively, in order to minimise the number of iterations of a "fast value iteration algorithm", we should seek for "potentials" assigning large values to vertices, so that  a ``big step'' is produced in each iteration. In this sense, if $\Psi \leq \Psi'$, the "$\Psi'$-FVI algorithm" is expected to perform better.
A priori, the sequence of games produced by the two algorithms will diverge, impeding formal comparisons on the number of iterations.
However, for ""monotone"" "potential assigners", we can also compare the number of iterations of the induced "FVI algorithms".
%Moreover, of course, an important factor for the efficiency of "$\Psi$-FVI" is the cost of computing $\Psi(\game)$.
%Nevertheless, we believe the comparison of "potential assigners" offers a valuable theoretical tool for comparing algorithms.

\begin{restatable}{lemma}{lemComparingPotentials}\label{lem:comparison-potentials}
	For every game $\game$, 
	\[\PsiFirst(\game) \leq  \PsiEnp(\game) \leq \PsiDPPI(\game), \text{ and } \PsiGKK(\game) \leq  \PsiEnp(\game).\]
	Moreover, there are games making these inequalities strict. The "potential assigners"	$\PsiFirst$ and $\PsiGKK$ are incomparable.
\end{restatable}

\AP Let $\Psi, \Psi'$ be two "potential assigners". We say that $\Psi'$ is ""monotonically larger"" than $\Psi$, noted $\Psi \intro*\leqMon \Psi'$ if, for all game and "potentials" $\phi \leq \phi'$, it holds 
\[ \phi + \Psi(\game_{\phi}) \leq  \phi' + \Psi'(\game_{\phi'}).  \]
\AP We say that $\Psi$ is ""monotone"" if $\Psi \leqMon \Psi$. The next two lemmas are immediate.

\begin{lemma}
	Let $\Psi, \Psi'$ be "sound", "complete" "potential assigners", and assume $\Psi \leqMon \Psi'$. Then over any input game $\game$, the "$\Psi'$-FVI" algorithm terminates in less iterations than the "$\Psi$-FVI" algorithm.
\end{lemma}

\begin{lemma}
	Let $\Psi, \Psi_1, \Psi_2$ be "potential assigners". It holds:
	\[ \Psi \leqMon \Psi_1 \leq \Psi_2 \; \implies \; \Psi \leqMon \Psi_2. \]
	In particular, if $\Psi$ is "monotone" and $\Psi \leq \Psi'$, then $\Psi \leqMon \Psi'$.
\end{lemma}

\begin{proposition}\label{prop:comparison-algorithms}
	The "potential assigner" $\PsiFirst$ is "monotone". Therefore, "PPI" and "DPPI" terminate in less iterations than "SVI" over any input game.
\end{proposition}
\iffull
\begin{proof}
	Let $\phi \leq \phi'$ be two "potentials" on $\game$. Let $v$ be a Min-vertex (the proof for Max-vertices is the same). It suffices to remark that:
	\[ \phi(v) + \PsiFirst(\game_{\phi})(v) = \min_{vv'\in E} \phi(v') + w(vv') \leq  \max_{vv'\in E} \phi'(v') + w(vv'). \]
\end{proof}
\else\fi

An interesting open question is whether the potential $\PsiGKK$ is monotone.

\section{Experimental results}\label{sec:empirical}

We focus on two distinct game-solving applications: energy game solving, which
is the natural target for our algorithms, and parity game solving, which incurs
a conversion cost to energy games but allows using established parity game
benchmarks and comparison with other parity game solvers.

After explaining the technical aspects of our implementation, and choices of algorithms and benchmarks, we discuss the most remarkable behaviours that can be observed in the experiments.

The algorithms were implemented in Oink~\cite{vanDijk18a}, a tool
providing a uniform framework for the comparison of parity game
solvers.  Our implementation can be obtained at: \href{https://github.com/michaelcadilhac/oink/tree/TACAS25}{\texttt{https://github.com/michaelcadilhac/oink/tree/TACAS25}}.

Experiments were carried on an Intel\textsuperscript{\textregistered}
Core\textsuperscript{\texttrademark} i7-8700 CPU @ 3.20GHz paired with 16GiB of memory, each
test being capped at 60 seconds and 10GiB of memory.
Arithmetic operations over multiple precision integers are carried out using the
GNU Multiple Precision Arithmetic library (GMP).  All games  are available at:
\href{https://github.com/michaelcadilhac/game-benchmarks/tree/TACAS25}{%
    \texttt{https://github.com/michaelcadilhac/game-benchmarks/tree/TACAS25}}.

%All the benchmarks used in this section, including the randomly generated ones, are available at: \url{http://github.com/michaelcadilhac/games-benchmark}.

\paragraph{Set of algorithms.}
We compare our implementations of "PPI", "DPPI" and their "alternating versions"
("PPI-alt" and "DPPI-alt")\footnote{In favour of clarity, we omit the "DPPI-alt"
  plots, as they perform identically to "PPI-alt". This is expected, given the
  similarity of the plots of "PPI" and "DPPI".} to $4$ other algorithms: QDPM
from~\cite{BDM24}, Zielonka's recursive algorithm (ZLK), Tangle learning (TL)
and Recursive tangle learning (RTL) (winner of the latest edition of SYNTCOMP)
from~\cite{vanDijk18b,vanDijk18a}.  Only one of them ("QDPM") can be executed
over general energy games, the other three are parity game solvers

%\begin{itemize}
%	\item QDPM,
%	\item ZLK (only parity),
%	\item TL (only parity),
%	\item RTL (only parity), the winner of the parity games track at SYNTCOMP24.
%\end{itemize}

We remark that we do not include comparisons with "SVI"~\cite{BCDGR11}, nor with "GKK-DKZ"~\cite{GKK88,DKZ19}, as these algorithms are known to be inefficient in practice~\cite{BDM24} and incur in frequent timeouts.
Also, we have not compared to an independent implementation of "OSI", as we have not found one such implementation computing winning regions consistent with the rest of the algorithms.

\subsection{Parity game solving}

We show the results of our experiments on parity games in
Figure~\ref{fig:parity-plots}.  We rely on the yearly competition SYNTCOMP24 for
our benchmarks, which has a competition track for parity game solvers, and on
the benchmark suite of Keiren~\cite{keiren15}.  We subdivide the 779 benchmarks
into two categories: synthetic games (crafted by researchers, usually with the
intent of being hard for certain solving approaches) and organic
games (the natural counterpart of the synthetic games).  We note that the
synthetic games include the ``two counter games''
examples~\cite{vD19TwoCounters}, in which TL and RTL show an exponential
behaviour. It also contains the family of examples by
Friedmann~\cite{Friedmann11b}, exponential for "OSI". \iffull (We refer to
  Appendix~\ref{app:friedmann} for more details on Friedmann's family of
  examples.)  \else\fi The organic games are essentially the ones provided by
Keiren~\cite{keiren15}, see therein for their origin.
%\begin{itemize}
%	\item \emph{Synthetic games:} There are games that were crafted by researchers, usually with the intent of being hard for certain parity game solving
%	approaches.  
%%	The term ``synthetic'' is used to stress the fact that they do 	not appear organically, for instance in industrial applications. 
%	%These tests	usually are [FIXME] dense?sparse?high prio?low prio?
%	\item \emph{Organic games:} This is the natural counterpart of the synthetic
%	games.  
%	%Since they come for a diversity of backgrounds, we split them into two 	categories: those with \emph{low} priorities (below 10), and those with \emph{high} priorities (above 10).  
%	%This is justified by the fact that most 	parity game solving approaches have an exponential dependency on the highest 	priority --- this is the case in our approach, for instance.
%\end{itemize}

\begin{figure}[h]
	\begin{center}
		\includegraphics[width=\linewidth]{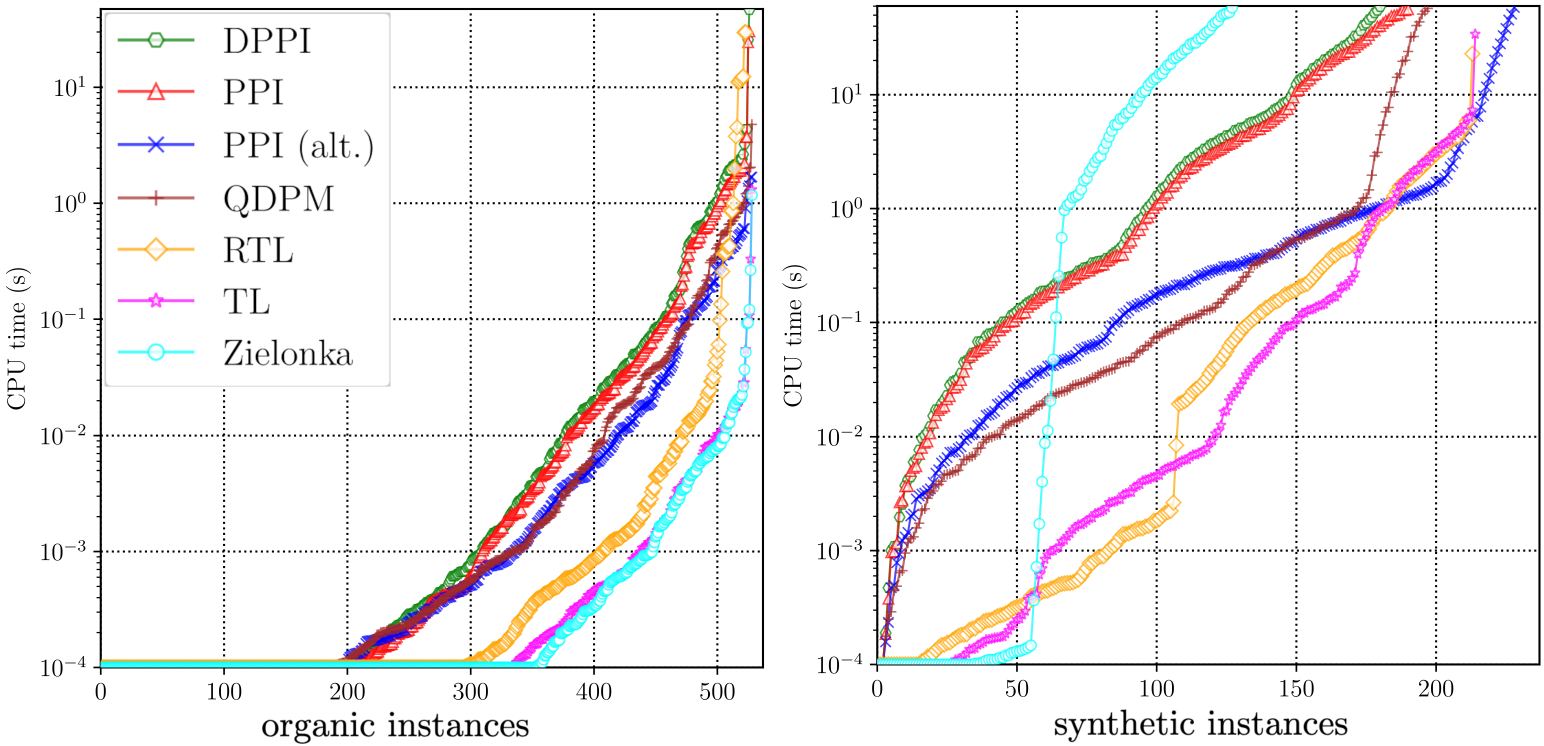}
	\end{center}
	\vspace{-4mm}
	\caption{Survival plot for parity games benchmarks, divided in organic and synthetic.}
	\vspace{-4mm}
	\label{fig:parity-plots}
\end{figure}

%\vspace{-3mm}

As is usual in this settings we present the experimental results as a survival plot, which indicates how many tests
are solved (x-axis) within a time limit (y-axis, time per test). % The lower the curve, the better.
In order to solve input parity games with energy games solvers, we
first need to convert the parity game into an energy one.
This step is rather costly, as the priorities of the parity game suffer an exponential blow-up when converted to weights of an energy game. This cost is included in the runtime of our algorithms as well as QDPM.

%\vspace{-3mm}

\subsection{Energy game solving}

We show the results of our experiments on parity games in Figure~\ref{fig:parity-plots}.
We modified Oink so that it would accept negative weights and implemented a
strategy-checker for energy games --- this boils down to checking that, in the
game restricted to the strategy, Max-winning strongly-connected components do
not have infinite negative cycles, and symmetrically for Min.

We consider randomly generated bipartite graphs. The restriction to bipartite
graphs is justified by the fact that, otherwise, the vast majority of vertices are part of winning cycles controlled by the same player, making the game (and its resolution) much easier.
We separate instances that are sparse (the out-degree of each vertex is 2) or dense (the number of edges is \(n^2 / 5\)).
%\begin{itemize}
%\item Either \emph{sparse} (the out-degree of each vertex is 2) or \emph{dense}
%  (the number of edges in a game with \(n\) vertices is \(n^2 / 5\)).
%\item Either \emph{low-weight} (weights are taken in the range \([-n, n]\) for
%  \(n\)-vertex graphs) or \emph{high-weight} (weights are taken in the range
%  \([-2^n, 2^n]\) for \(n\)-vertex graphs).  This is relevant on a theoretical
%  basis, since weights have an important impact on the complexity of the
%  algorithms, and on an applied basis, since transforming parity games into
%  energy games involves converting priorities to exponential weights.
%\end{itemize}

\begin{figure}[h]
	\begin{center}
		\includegraphics[width=\linewidth]{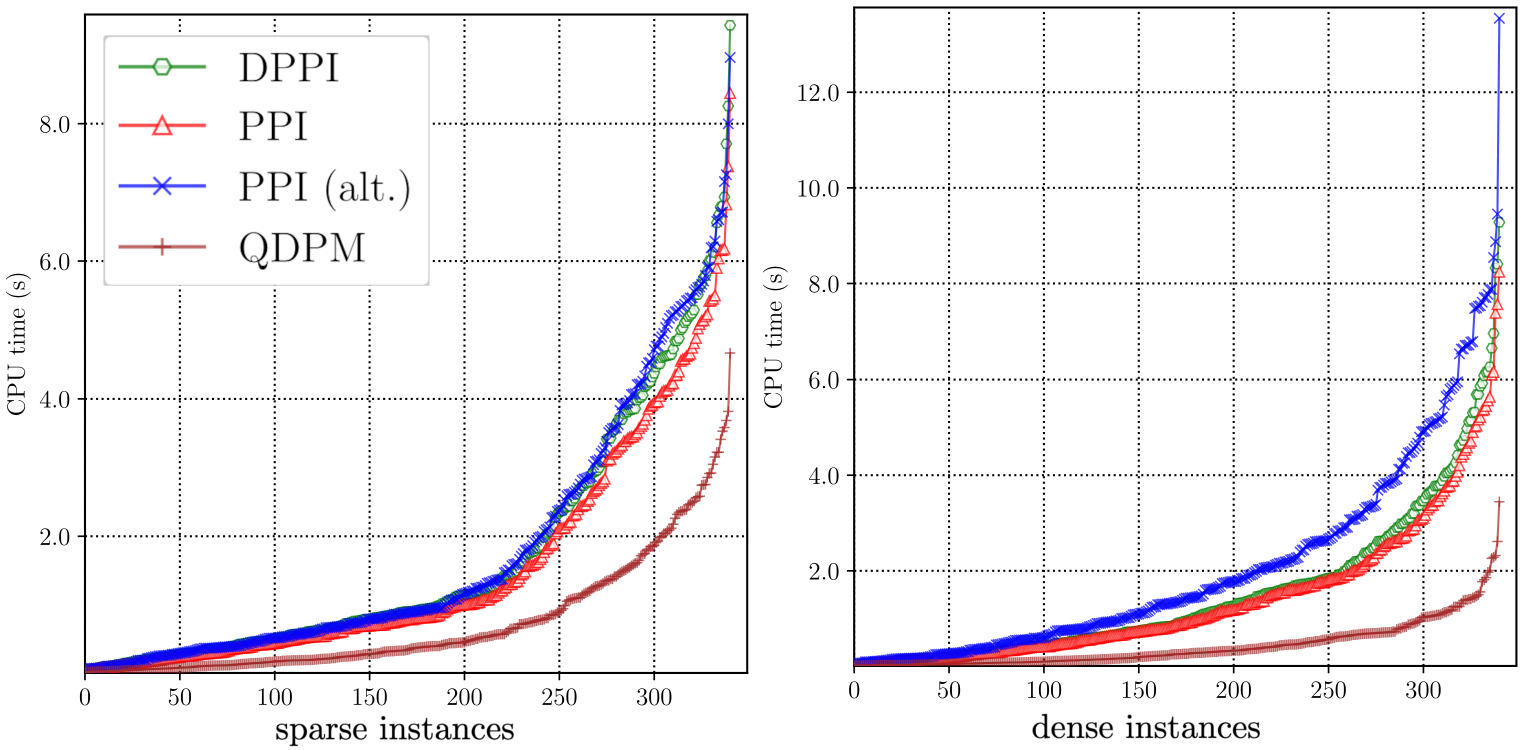}
	\end{center}
	\vspace{-4mm}
	\caption{Survival plot for energy games benchmarks, divided in sparse and dense.}
	\vspace{-4mm}
	\label{fig:energy-plots}
\end{figure}

%\vspace{-2mm}

\subsection{Conclusions of the experiments}
In light of the experiments above, we derive the following conclusions.
\begin{enumerate}
	\item Overall, the "fast value iteration" framework captures several algorithms ("PPI", "PPI-alt", "QDPM") that perform competitively in standard benchmarks of parity games. Despite being less efficient than leading parity game solvers, they are remarkably robust against hard instances, particularly "PPI-alt".
	\item The "alternating version" of "PPI" and "DPPI", for which we were unable to prove termination in theory, always terminate. Moreover, over instances coming from parity games benchmarks, they are significantly faster than their asymmetric counterparts.
	\item While "DPPI" was introduced as a theoretically enhanced version of "PPI", there is no significant difference in the running time of these algorithms. In fact, "DPPI" tends to be slightly slower, due to the increased cost in the computation of the potential.
	\item Although based on the same algorithmic ideas, "QDPM" consistently outperforms "PPI", by almost an order of magnitude.
	This difference can be explained by two factors: (1) "QDPM" uses some smart implementation optimizations~\cite[Sect.~5]{BDM24}, and 
	(2) our implementation of "PPI" is tailored for (usual) edge-weighted games, whereas "QDPM" is implemented for vertex-labelled game (for which two weights outgoing a given vertex are always equal).
	\iffull
	Details and estimates on why and how this difference may affect the performance of the algorithms is discussed in detail in Appendix~\ref{app:implementationDifferences}.
	\else\fi
	
\end{enumerate}

\section{Conclusion and future work}
We have presented a general framework to describe algorithms for energy games, capturing and providing simple descriptions and correctness proofs for many of them, including the top performing ones in practice.
The "fast value iteration" framework raises numerous exciting questions; we outline some of them here.

\paragraph{New algorithms.} The new framework provides a very easy way to propose new correct algorithms: it suffices to define a "potential assigner" which is "sound", "complete", and computable in polynomial time.
We have isolated $\Psi_{\Enp}$ as a important potential assigner, implicitly used by the two fastest algorithms solving energy games, and presented the potential $\Psi_\DPPI$ which, while still being computable in polynomial time, is $\geq \Psi_{\Enp}$ in general.

\begin{question}\label{quest:newAssigners}
	Does there exist a reasonable\footnote{A non-reasonable example meeting the requirements is $\Psi_{\DPPI} \circ \Psi_\DPPI$.} potential assigner which is "sound", "complete", computable in polynomial time and $\geq \Psi_\DPPI$?
\end{question}

%Of course we expect many positive answers to Question~\ref{quest:newAssigners}, which is an exciting perspective for reasearch.

\vspace{-2mm}
\paragraph{Alternating algorithms.} Our framework also allows to design symmetric "alternating algorithms", for which we are unable to prove termination using the currently available tools.
Our empirical study shows that, in practice, these not only terminate, but are often considerably faster than their asymmetric counterparts.

\begin{question}
	Do "alternating fast value iterations" terminate over "simple games"?
\end{question}

We stress the fact that the question is open for all "sound" and "complete" "potential assigners" (except for "GKK", for which the "alternating algorithm" coincides with the normal one, see Section~\ref{subsec:GKK-DKZ}).

\paragraph{Lower bounds.} 
Friedmann proposed notoriously involved constructions which provide exponentially many iterations for strategy improvement algorithms in~\cite{Friedmann11b}. \iffull
(We discuss in detail Friedmann's family of examples in Appendix~\ref{app:friedmann}.)
\else\fi
Although  these include "OSI" (see~\cite[Sect.~4.6.2]{Friedmann11b}), our experiments show that "PPI" can solve these instances in linear time, and "PPI-alt" in a constant number of $2$ iterations.
Currently, we lack any family of examples in which "PPI" takes more than a linear number of iterations, although we expect that it should admit exponential lower bounds.
%Our experiments show that, although as expected OSI and also QPDM (as implemented by~\cite{}) indeed require exponentially many iterations, this is not the case of PPI.
%Although we expect that PPI should admit similar counterexamples --Friedmann's examples have proved robust enough to adapt to a plethora of algorithms-- we wonder whether the simplified presentation we gave could lead to simplified hard instances.

\begin{question}
Can one design superpolynomial lower bounds on the number of iterations for "PPI"? And (more challenging) for its "alternating variant"?
\end{question}

\vspace{-2mm}
\paragraph{Randomized initialization.} As remarked in Section~\ref{subsec:potential-reductions}, the weights of the cycles of $\game$ and $\Gphi{\phi}$ coincide for any finite "potential" $\phi$, so the threshold problem for the mean-payoff objective is equivalent over these games.
Therefore, we can initialize a given game with an arbitrarily potential $\phi$, and solve the ``perturbed game''. This directly provides a randomized version of any algorithm: add a random perturbation before execution.
This idea is not novel, it was studied empirically by Beffara and Vorobyov~\cite{BV01} for the "GKK algorithm"; and lower bounds were later derived by Lebedev~\cite{Lebedev16} for the same algorithm.

\begin{question}
	Is the randomized variant of "PPI" subexponential? More generally, can we design a potential assigner whose associated randomized fast value iteration is subexponential?
\end{question}

\vspace{-2mm}

\paragraph*{Smooth analysis.}
An interesting parallel can be drawn with smooth analysis~\cite{ST04SmoothAnalysis}, which consider small perturbations of the input (randomized initialization is difference since we get an equivalent input). In fact, it was recently established that there is a strategy improvement algorithm for mean-payoff games that is polynomial in the sense of smooth analysis~\cite{LS24SmoothedAnalysisMP}.

\begin{question}
	Can the algorithm of~\cite{LS24SmoothedAnalysisMP} be recast as a fast value iteration?
\end{question}

\paragraph*{Acknowledgments.}  We thank the authors of \cite{BDM24} for
kindly providing their implementation of QDPM and Alexander Kozachinskyi for pointing out to us several important references.\\
Antonio Casares is supported by the Polish National Science Centre (NCN) grant ``Polynomial finite state computation'' (2022/46/A/ST6/00072).
%%
%% Bibliography
%%

 \bibliographystyle{splncs04}
 \bibliography{bib}
 
 \iffull
 \newpage
 \appendix
\section{Correctness of algorithms}
	 \subsection{Correctness of the fast value iteration meta-algorithm}\label{app:metaAlg-correctness}
	 First, observe that for a finite path $\pi= v_0 \re v_1 \re \dots \re v_k$ which visits only vertices with finite potential, its sum in $\game_\phi$ is given by
\[
\summ_\phi(\pi) = \summ(\pi) - \phi(v_0) + \phi(v_k).
\]
We start with a technical lemma.

\begin{lemma}\label{lem:optimal_strat}
	Let $\sigma_0$ be an $\En$-optimal Min strategy in $\game$ and $\pi = v_0 \re v_1 \re \dots \re v_k$ be a finite path consistent with $\sigma_0$ such that $\EnG(v_k)<\infty$.
	Then we have $
	\summ(\pi) \leq \EnG(v_0) - \EnG(v_k)
	$.
\end{lemma}

\begin{proof}
	Let $\pi'$ be an infinite path from $v_k$ consistent with $\sigma_0$ and such that $\EnG(v_k)=\En(w(\pi'))$.
	Then $\pi \pi'$ is consistent with $\sigma_0$ thus $\EnG(v_0) \geq \En(w(\pi \pi'))$ by optimality.
	We then obtain 
	\[\arraycolsep=1.4pt\def\arraystretch{1.3}
	\begin{array}{lcl}
		\EnG(v_0) \ \ \geq \ \ \En(w(\pi\pi')) & = & \sup_{k' \geq 0}(\summ((\pi\pi')_{<k'})) \\ & \geq & \sup_{k' \geq k} (\summ((\pi \pi')_{<k'})) \\
		& = & \summ(\pi) + \sup_{k' \geq 0} \summ(\pi'_{<k'}) \\
		& = & \summ(\pi) + \En(w(\pi')) \ \  =  \ \ \summ(\pi) + \EnG(v_k). \ \ %\quad \qed
	\end{array}
	\]
	
\end{proof}

We are now ready to prove Theorem~\ref{thm:potential_reduction_theorem}, which we first restate for convenience.

\potentialReduction*

\begin{proof}%[Proof of Theorem~\ref{thm:potential_reduction_theorem}]
	Let $\phi: V \to \Ninfty$ be a potential such that $\phi \leq \EnG$; we aim to prove that $\EnG = \phi + \En_{\game_\phi}$ over $V$.
	Consider first a vertex $v$ with $\EnG(v)=\infty$, fix an optimal Max strategy $\tau_0$ in $\game$ and an infinite path $\pi=e_0e_1 \dots = v_0 \re v_1 \re \dots$ consistent with $\tau_0$ from $v$: by definition we have $\En(w(\pi)) = \sup_k \sum_{i=0}^{k-1} w(e_i) = \infty$.
	We claim that $\En(w_\phi(\pi))=\infty$ which proves the wanted equality over $v$ (both terms are infinite).
	\begin{itemize}
		\item If for some $i$, $w(e_i) = \infty$ then $w_\phi(e_i) = \infty$ which implies the result.
		\item If for some $i$, $\phi(v_i)=\infty$ then again we have $w_\phi(e_i)=\infty$.
		\item Otherwise, we have for all $k$
		\[
		\summ_\phi(\pi_{<k}) = \underbrace{\phi(v_k) - \phi(v_0)}_{\text{bounded}} + \summ(\pi_{<k}),
		\]
		and therefore $\sup_k \summ_\phi(\pi_{<k}) = \sup_k \summ(\pi_{<k}) = \infty$, the wanted result.
	\end{itemize}
	
	We now consider a vertex $v$ such that $\EnG(v) < \infty$.
	Consider an $\En$-optimal Min strategy $\sigma_0: \VMin \to E$ in $\game$ and let $\pi=v_0 \re v_1 \re \dots$ be an infinite path consistent with $\sigma_0$ starting from $v_0=v$.
	Note that for any $k \geq 0$, $v_k$ has finite energy value, and thus we obtain thanks to Lemma~\ref{lem:optimal_strat} and the hypothesis $\phi \leq \EnG$ that
	\[\arraycolsep=1.4pt\def\arraystretch{1.3}
	\begin{array}{lcl}
		\summ_{\phi}(\pi_{<k}) & = & \summ(\pi_{<k}) + \phi(v_k) - \phi(v_0)\\
		& \leq & \EnG(v_0)  \underbrace{- \EnG(v_k) + \phi(v_k)}_{\leq 0} - \phi(v_0) \ \ \leq \ \ \EnG(v_0) - \phi(v_0),
	\end{array}
	\]
	hence $\En_{\phi}(v_0) = \sup_{\pi \const \sigma_0} \sup_{k \geq 0} \summ_\phi(\pi_{<k}) \leq \EnG(v_0) - \phi(v_0)$.
	
	For the other inequality, consider an optimal Min strategy $\sigma_\phi$ in $\game_\phi$, and let $\pi$ be an infinite path from $v_0=v$ consistent with $\sigma_\phi$.
	By applying Lemma~\ref{lem:optimal_strat} in $\game_\phi$ we now get
	\[\arraycolsep=1.4pt\def\arraystretch{1.3}
	\begin{array}{lcl}
		\summ(\pi_{<k}) & = & \summ_\phi(\pi_{<k}) - \phi(v_k) + \phi(v_0)\\
		& \leq & \En_{\game_\phi}(v_0) - \underbrace{\En_{\game_\phi}(v_k)}_{\geq 0} - \underbrace{\phi(v_k)}_{\geq 0} + \phi(v_0) \ \ \leq \ \ \En_{\game_\phi}(v_0) + \phi(v_0),
	\end{array}
	\] 
	and the wanted result follows by taking a supremum.
\end{proof}

When considering a map $\phi:S \to \Ninfty$ defined over a subset $S$ of the vertices, we let $\overline \phi$ denote the extension of $\phi$ to $V$ defined by setting its value to be $\infty$ on $\compl S$.
We require an additional technical lemma.

\begin{lemma}\label{lem:remove-attr}
	Let $\game$ be a game, let $S$ be a set of vertices with energy value $\infty$ in $\game$, let $A=\Attr_\game^\Max(S)$ and let $\game'=\game \setminus A$.
	Then $\EnG(A) =\infty$ and $\overline{\En_{\game'}} \leq \EnG$.
\end{lemma}

In fact, the last inequality is even an equality, but we only require this direction.

\begin{proof}
	By a strategy forcing to first go to $S$, then following an $\En$-optimal strategy, Max ensures energy-value $\infty$ over $A$.
	In particular, we get that $\overline{\En_{\game'}}$ and $\EnG$ both have value $\infty$ over $A$.
	Now since $A$ is a Max-attractor, any Min strategy over $\compl A$ forces the game to stay in $\compl A$ therefore $\En_{\game'} \leq \EnG|_{\compl A}$, which concludes the proof.
\end{proof}

We are now ready to prove Theorem~\ref{thm:potential-reduction-algorithm}.
The proof requires a bit of bookkeeping regarding vertices sent to $\infty$.

\metaAlg*

\begin{proof}
	Let $\game_0=\game$, $A_0=\emptyset$, $\Phi_0=0$ and for each iteration $i=1,2,\dots$ let $\game_i,\Phi_i$ and $\psi_i$ be the values of the corresponding variables after line~\ref{line:update-cumulative}, and let $A_i$ be computed on line~\ref{line:remove-attractor}.
	Let $V_i$ be the vertex set of $\game_i$ so that $\compl{V_i}=A_0 \cup \dots \cup A_{i-1}$.
	By definition, we have 
	\[
			\phi_i=\Psi(\game{i}), \quad
			\Phi_i=\Phi_{i-1} + \overline{\phi_i}, \quad \text{and} \quad
			\game_i=(\game_{i-1})_{\phi_{i-1}} \setminus A_{i-1}.
	\]
	We prove by induction on $i$ that for all $i$, $\Phi_{i}$ is "sound" for $\game$ and that $\EnG(A_i) \subseteq \{\infty\}$.
	For $i=0$ there is nothing to prove, so we let $i > 0$ and assume the result known for $j< i$.
	
	Let $\overline {\game_i}$ denote $\game_{\Psi_{i-1}}$.
	Then $\game_i = \overline{\game_i} \setminus \Attr_\game^\Max(\compl{V_i})$ and we know by induction that $\En_G$ is $\infty$ over $\compl V = A_0 \cup \dots \cup A_{i-1}$.
	Thus Lemma~\ref{lem:remove-attr} gives $\overline{\En_{\game_i}} \leq \En_{\overline \game_i}$.
	Now since $\Psi$ is "sound" we have $\phi_i \leq \En_{\game_i}$ which implies $\overline{\phi_i} \leq \overline{\En_{\game_i}} = \En_{\overline{\game_i}}$, and thus $\overline \phi_i$ is "sound" for $\overline{\game_i}$.
	Thanks to the induction hypothesis and compositionality we deduce that $\Phi_i = \Phi_{i-1} + \phi_i$ is "sound" for $\game$.
	
	There remains to prove that $\EnG(A_i) \subseteq \{\infty\}$.
	Let $v \in \Psi_i^{-1}([(n-1)W+1,\infty])$.
	Then $\EnG(v) \geq \Psi_i(v)$ therefore by Corollary~\ref{cor:mean_payoffs_and_energies} we get $\EnG(v)=\infty$, and conclude thanks to Lemma~\ref{line:remove-attractor}.
\end{proof}

	 \subsection{A sufficient condition for "$\infty$-attraction" }\label{app:infAttracting}
	 
   Recall that $N_\game$ is the set of vertices such that Min can force to immediately see a negative edge.
We say that a potential assigner is $N$-null if for any $\game$, $\Psi(\game)(N_\game) = 0$, that it is $\Psi$ is $N$-shrinking if for all $\game$ it satisfies $N_{\Psi(\game)} \subseteq N_\game$, and that it is path-based if finite values of $\Psi(\game)(v)$ are upper bounded by the weight of a simple path from $v$ to a vertex in $N_\game$.

\begin{theorem}\label{thm:infty-attr}
	Any potential assigner which is path-based and $N$-shrinking is "$\infty$-attracting".
\end{theorem}

\begin{proof}
	Consider such a potential assigner $\Psi$ and a game $\game_0$.
	Forall $i$, let $\phi_i=\Psi(\game_i)$, $\game_{i+1}=(\game_i)_{\phi_i}$ and $\Phi_i=\phi_0 + \dots + \phi_{i-1}$.
%	We show that the values of $\Phi_i$ are either $\leq (n-1)W$ or $\infty$, which implies the result.

	Since $\Psi$ is $N$-shrinking, we get (with obvious notations) $N_0 \supseteq N_1 \supseteq \dots$ and therefore since it is moreover $N$-null, vertices $v'$ in $N_i$ satisfy $\Phi_{i}(v')=0$.
	Now if $v$ is a vertex such that $\phi_i(v)$ is finite, then since $\Psi$ is path based, there is a simple path $\pi=v_0 \re \dots \re v_k=v' \in N_i$ in $\game_0$ from $v$ whose $\Phi_j$-modified sum satisfies $\summ_{\Phi_j}(\pi) \geq \phi_j(v)$.
This rewrites as
\[
	\underbrace{\phi_j(v)}_{\geq 0} \leq - \Phi_j(v) + \underbrace{\Phi_j(v')}_0 + \underbrace{\sum_{i=0}^{k-1} w_{\Phi_{j}}(v_iv_{i+1})}_{\leq (n-1)W} ,
\]
and thus $\Phi_j(v) \leq (n-1)W$.
Stated differently, finite values remain $\leq (n-1)W$, which guarantees termination in at most $O(n^2W)$ iterations.
\end{proof}

There remains to see that $\Psi_{\Enp}$, $\Psi_{\DPPI}$ and $\PsiGKK$ satisfy the hypotheses of Theorem~\ref{thm:infty-attr}.
It is obvious that they are $N$-null and path-based.

\begin{lemma}\label{lemma:N-decreases}
	The potential assigners $\Psi_{\Enp}$, $\Psi_{\DPPI}$ and $\PsiGKK$ are $N$-shrinking.
\end{lemma}

\begin{proof}
	Let $v \notin N_\game$.
	Then observe that for all $X \in \{\Enp, \DPPI, \GKK\}$, it holds that if $v \in \VMax$ (resp. $v \in \VMin$) then for some (all) successors $v'$ it holds that $ \Psi_X(v) \geq w(vv') + \Psi_X(v)$.
	This is the same as saying that Max can ensure that a non-negative weight is immediately seen from $v$ in the $\Psi_X$-modified game, that is $v \notin N_{\game_{\Psi_X}}$.
\end{proof}

   \subsection{Correctness and complexity of Algorithms~\ref{alg:PPI_OSI} and~\ref{alg:PPI_QD}}\label{app:correctnessPPI}
   
We now prove Theorem~\ref{thm:correctnessPPI}.

\correctnessPPI*

\begin{proof}
We prove correctness of both algorithms using a similar induction, stating that $\Phi$ coincides with $\Enp$ over $\compl{F}$.
In both cases this is true when the while loop starts, since $\Enp$ is zero over $N_\game$, so we focus on the inductions step.
The two proofs below very are based on similar ideas, we separate them for clarity.

\vskip1em

\emph{Induction step for Algorithm~\ref{alg:PPI_OSI}.} There are two cases.
\begin{itemize}
	\item If $\finEscMax \neq \emptyset$. Let $v \in \finEscMax$ be chosen by the algorithm.
	Since $\escF(v)<\infty$, all positive edges outgoing from $v$ lead to $\compl{F}$.
	Now an optimal $\Enp$ strategy from $v$ should surely start with a positive edge, say, going to $v' \in \compl F$.
	This concludes, since by induction, $\Enp$ and $\Phi$ coincide over $\compl{F}$.
	\item If $\finEscMax = \emptyset$. Then let $v \in \finEscMin$ be chosen by the algorithm, and $v'$ be such that $m=w(vv')+\Phi(v')$.
	We claim that the edge $vv'$ is $\Enp$-optimal, which proves the wanted result by induction.
	Indeed, an other edge $vv''$ that ends in $\compl{F}$ would lead to value $w(vv'') + \Phi(v'') \geq m$ by minimality.
	Now if Min plays an edge towards $F$, Max can force the game to remain in $F$ while visiting only non-negative edges (since $\finEscMax =\emptyset$).
	Therefore such a play remains in $F$ until potentially going to $\compl F$ via a Min vertex, and thus its value is $\geq m$ by induction. 
\end{itemize}

\emph{Induction step for Algorithm~\ref{alg:PPI_QD}.} Let $v$ be the vertex chosen by the algorithm, meaning $\escF(v)$ is minimal, and in particular it is finite.
Let $\game'$ be the modified game at this stage of the algorithm, note that $\game'=\game_{\Phi}$.
\begin{itemize}
	\item If $v \in \VMax$.
	Since $\escF(v)<\infty$, all positive edges in $\game'$ outgoing from $v$ lead to $\compl{F}$.
	But since $\Phi(v)=0$, positive edges outgoing from $v$ in $\game'$ are also positive in $\game$.
	
	Now an optimal $\Enp$ strategy from $v$ in $\game$ should surely start with a positive edge.
	The modified weights in $\game'$ of edges from $v$ to $v' \in \compl{F}$ are of the form $w(vv') + \Phi(v')=w(vv')+\Enp(v')$ by induction.
	We conclude that the edge $vv'$ maximising $\escF(v)$ satisfies $\escF(vv') = \Enp(v)$ which is also the final value of $\Phi(v)$.
	
	\item If $v \in \VMin$.
	Then let $v'$ be such that $\escF(v)=w_\Phi(vv')$.
	We claim that the edge $vv'$ is $\Enp$-optimal, which proves the wanted result by induction; starting with edge $vv'$ then playing optimally gives value $w(vv')+\Enp(v')$, which correspond to $w_\Phi(vv')=\escF(v)$ by induction. 
	First, an other edge $vv''$ that ends in $\compl{F}$ would lead to value $\escF(vv'')$ by the same argument, which is $\geq \escF(vv')$ (hence, less optimal for Min) by minimality.
	Now if Min plays an edge towards $F$, Max can force the game to either remain in $F$ while visiting only non-negative edges in $\game'$ (since $\finEscMax =\emptyset$), or leave towards $\compl{F}$ via an edge $uu'$ with modified weight $\escF(uu') \geq \escF(vv')$.
	But since $\Phi$ is $0$ over $F$, non-negative edges in $\game'$ are also non-negative in $\game$, which concludes.
\end{itemize}

In both cases, the while loop terminates when $\finEsc =\emptyset$.
This means that Max can ensure that plays starting in $F$ visit only $\geq 0$ vertices.
But since $\game$ is simple, this implies that $\Enp$ is indeed $\infty$ over $F$, and thus it coincides with $\Phi$ everywhere.

Updating values of $\finEsc$ requires, as is standard in such game algorithms, storing the number of positive edges, from each Max vertices towards $F$, and updating predecessors of vertices added to $F$.
This incurs a runtime of $O(m)$.
For algorithm~\ref{alg:PPI_QD}, updating minimal value of $\escF$ requires using a priority queue, and the same technique can be applied in algorithm~\ref{alg:PPI_OSI} to maintain the value of $m$.
This induces a runtime of $O(n\log n)$, just like in Dijkstra's algorithm~\cite{FT84}.
\end{proof}
	 \subsection{Soundness and completeness of $\PsiDPPI$}\label{app:correctnessDPPI}
	 \correctnessDPPI*
\begin{proof}
	Fix a "simple game" $\game$.
	We start proving "soundness" ($\PsiDPPI(\game) \leq \EnG$).
	Let $\game_j,F_j,v_j$ and $\phi_j$ be, respectively, the values of $\game,F,v$ and $\phi$ after line~\ref{line:endIf} at the $j$th iteration of the algorithm, and let $\Phi_j=\phi_1 + \dots + \phi_{j-1}$, so that $\game_j= \game_{\Phi_j}$.
	Thanks to compositionality (Corollary~\ref{cor:compositionality}), it suffices to prove that $\phi_j$ is "sound" in $\game_j$, so we should prove that $\phi_j(v_j)=\escFj(v) \leq \EnGame{\game_j}(v_j)$.
	
	Note that $\Phi_j$ is $0$ over $F_j$ so edges outgoing from vertices in $F_j$ have a weight in $\game_j$ greater or equal to their weight in $\game$.
	In particular, $F_j \subseteq V\setminus{\negG{\game_j}}$, hence Max can ensure that only edges with non-negative weights are seen over $F_j$.
	Consider the following Max-strategy defined over $F_j$: if there is a non-negative edge towards $F_j$ from the current vertex $v$, play it; otherwise play an edge maximising $\escF(v)$.
	%Over $\compl{F_j}$, Max may play arbitrarily.
	We claim that this strategy achieves $\En$-value $\geq \escF(v)$ in $\game_j$.
	Consider a play $\pi$ from $v_j$ consistent with the strategy; there are two cases.
	\begin{itemize}
		\item If $\pi$ remains in $F_j$, then only non-negative weights are seen, and therefore since the game is simple, the value of the play is $\infty \geq \escFj(v)$.
		\item Otherwise, $\pi$ visits only non-negative edges within $F_j$ until following an edge $F_j \ni v \re v' \notin {F_j}$.
		Then the weight of this edge in $\game_j$ is $\geq \escFj(v)$, which concludes.
	\end{itemize}
	
	Finally, note that after the execution of the \texttt{while}-loop (line~\ref{line:afterWhile-DPPI}) $F = V\setminus \Attr_{\game}(\negG{\game})$. Therefore, $\EnG(v) = \infty$ for all those vertices.
	
	We prove "completeness", that is $\EnG \neq 0 \implies \PsiDPPI(\game)>0$. If $\EnG>0$, then there is some vertex $v$ from which "Max can immediately see" a positive weight;
	note that $\escF(v)>0$ for all $F\subseteq V$ containing this vertex.	
	Let $v_j$ correspond to the first such vertex encountered by the algorithm. Note that $v_j \notin \negG{\game}$  and since $\Phi_j=0$ we have $\escFgame{\game_j}(v_j)=\escFgame{\game}(v_j)>0$.
	Therefore $\phi_j(v_j)>0$ hence $\Psi(\game)>0$.
\end{proof}

 \section{Comparisons between related algorithms}\label{app:comparisons}
 \subsection{Comparison of potential assigners}\label{app-sec:comparingPotentials}

\lemComparingPotentials*
\begin{proof}
	We focus on the proof of the inequalities, and discuss below examples separating the different potentials.
	
	($\PsiFirst(\game)\leq \PsiEnp(\game)$.) Follows directly from the fact that $\First \leq \Enp$ over sequences of weights.	
	
	($\PsiGKK(\game)\leq \PsiEnp(\game)$.) By definition of $\PsiGKK$, we have $\PsiGKK(\game) \leq \wplus \leq \EnpG$.
	
	($\PsiEnp(\game)\leq \PsiDPPI(\game)$.)
	Let $\Phi = \PsiDPPI(\game)$.
	Let $v$ be a vertex assigned potential $\Phi(v) = x <\infty$ in the $j$th iteration of Algorithm~\ref{alg:DPPI} (the property trivially holds for $v$ with $\Phi(v) = \infty$). Assume by induction that $\Enp(v') \leq \Phi(v')$ for all $v'$ that have been treated previously.
	%	 (with $\phi(v')$ finite by Remark~\ref{rmk:monotonicity-QD}).
	%	Let $F_j$ the set $F$ at this iteration; note that weights between vertices in $F_j$ coincide with those in the original game.
	Let $v'\notin F$ such that $vv'\in E$ is the transition determining $\Phi(v)$, that is, $\Phi(v) = w(vv') + \Phi(v')$. We distinguish two cases according to the player controlling $v$. If $v\in \VMin$, then:
	\[ \Enp(v) \leq w(vv') + \Enp(v') \leq  w(vv') + \phi(v') = \phi(v),\]
	where the second inequality follows by induction hypothesis.
	
	If $v\in \VMax$, let $u$ be a successor of $v$ such that $\Enp(v) = w(vu) + \Enp(u)$. Note that $u\notin F$, as otherwise $\phi(v) \geq \escF (vu) = \infty$. Therefore:
	\[ \Enp(v) = w(vu) + \Enp(u) \leq  w(vu) + \phi(u)  \leq w(vv') + \phi(v') = \phi(v),\]
	where the second inequality follows by induction, and the third one because $vv'$ is the edge maximizing the escape weight from $F$ in $\game$.
\end{proof}

A game separating PPI and DPPI was given in Figure~\ref{fig:exampleDPPI}.

%\ac{todo, if possible: Figure of game with 3 vertices}
%\begin{figure}[ht]
%	\begin{center}
%		\includegraphics[width=\linewidth]{Fig/???}
%	\end{center}
%	\caption{Game in which $\PsiEnp(\game)<\PsiDPPI(\game).}
%	\label{fig:Enp-vs-DPPI}
%\end{figure}

\begin{example}[$\PsiFirst < \PsiGKK$]
	Consider the game with a single vertex $v$ and two self loops, with weights $1$ and $W$. We have that $\PsiFirst(\game)(v) = 1$, and "SVI" takes $W$ iterations to realize that $\EnG(v) = \infty$. However, $\PsiGKK(\game)(v) = \PsiEnp(\game)(v) \PsiDPPI(\game)(v) = \infty$; all the other algorithms terminate in a single iteration.
\end{example}

\begin{example}[$\PsiGKK < \PsiFirst$ is slow]
	We note that for all game $\game$, the image of $\PsiGKK(\game)$ contains at most two value: $0$ and $\wplus$ (or $\wminus$).
	Consider the game with three vertices controlled by Max $V = \VMax = \{v_1,v_2,v_N\}$, and edges given by: $v_N \rew{-1} v_N$, $v_1 \rew{1} v_N$, $v_2 \rew{2} v_N$.
	Then, we have $\wplus = \wminus = 1$, and  $\PsiGKK{\game}(v_1) = \PsiGKK{\game}(v_2) =1$. The "GKK" algorithm takes 2 iterations to solve this game.
	However, $\PsiFirst{\game}(v_2) = 2$, and "SVI" takes a single iteration to solve the game.
\end{example}

\subsection{Comparison between PPI and OSI}\label{app-sec:PPI-vs-OSI}

We now describe the algorithm OSI, explaining the similarities and differences with PPI, our presentation within the fast value iteration framework.\footnote{Note that OSI was originally presented exclusively over parity games. It can easily be generalized to energy games, in the following we will always refer to this straightforward generalization.}\textsuperscript{,}\footnote{We note that Player 0 in~\cite{Schewe08,Luttenberger08} corresponds to our player $\Max$.}%, and Player 1 to $\Min$.}

OSI relies on the notion of \emph{estimations}, which correspond to our potentials.
To update an estimation (basic update in~\cite[p.377]{Schewe08}), OSI uses the auxiliary \emph{update game} $\mathcal{E}_\phi$, obtained from $\game_{\phi}$ by: i) adding a sink state $\bot$ to which $\Max$ can retreat at any point, ii) $\Max$-choices are restricted to non-negative edges. In this game, $\Max$ tries to maximize the weight of a play before reaching $\bot$. That value almost coincides with the $\Enp$-values of $\game_{\phi}$, and the subroutine used to solve the update game is very similar to the Algorithm~\ref{alg:PPI_OSI}; the main difference is that in $\Enp$ we stop the game as soon as one of the players produces a negative edge.
Due to this difference, some extra technical steps are required in the presentation of OSI:
\begin{itemize}
	\item The presentation of the algorithm is restricted to bipartite graphs.
	\item An initialisation step in which a first potential $\phi_0\colon V \to \N$ is computed is required. In the case of a bipartite graph, these are just the $\First$-values of $\Min$-vertices.
	\item Before each basic update state, we need to ensure that $\Min$ will not have the opportunity to visit negative edges in the update game. To this end, the current potential $\phi$ needs to be decreased in some $\Max$-positions (point 2 at the bottom of~\cite[p.379]{Schewe08}).
\end{itemize}

It is worth mentioning that, soon after the introduction of Schewe's algorithm, Luttenberger~\cite{Luttenberger08} proposed a reformulation as an explicit switching policy in the strategy improvement framework.
Although a potential is still used to guide the updates of the strategies, it comes organically as the evaluation of the current strategies.
To compute this evaluation, Luttenberger uses an adaptation of the Bellman-Ford algorithm, which is less efficient than Dijkstra's.

%Luttenberger's algorithm considers nondeterministic strategies, where Player $\Max$ selects multiple possible successors at each step, leaving the opponent free to choose any of these moves during a play. 
%As in Schewe and Bj\"orklund and Vorobiov, an extra sink where $\Max$ can retreat at any point is required, and finite plays ending in this sink are allowed.
%Each update consists in the ``all-switch'' policy; the new $\Max$'s strategy contains all the moves that led to a better evaluation.

For the reasons stated above, we see PPI as a polished and streamlined version of Schewe's algorithm.
In particular, PPI avoids the introduction of an additional sink vertex, answering a question by Bj\"orklund and Vorobiov~\cite[Conclusion]{BV05}.
The discrepancies on the running time (see Section~\ref{sec:empirical}) can be explained by (1) the extra computation steps that appear in the original description of OSI, and (2) the initialization to a slightly different potential in OSI's first step.

\subsection{Comparison between QDPM and PPI}\label{app-sec:QDPM-vs-FVIQD}
We now describe the algorithm "QDPM", explaining the similarities and differences with "PPI", our presentation within the "fast value iteration" framework.

The presentation of QDPM from~\cite{BDM24} is based on the notion of quasi dominions. 
A subset of positions $Q\subseteq V$ is a quasi dominion if player $\Max$ has a strategy ensuring to visit only non-negative weight as long as the play does not exit $Q$.
Therefore, player $\Min$ has an incentive to leave such a region as soon as possible.
We observe that the set $F = \compl{\negG{\game}}$ from which $\Min$ cannot force to immediately see a negative edge is a quasi dominion in the game $\game$.
The algorithm "PPI" finds a strategy for $\Min$ to leave this quasi dominion minimising the energy.

The main iteration principle of "QDPM" is provided by the operator $\mathsf{prg}_+$ (\cite[Alg.~1]{BDM24}). This almost corresponds to Algorithm~\ref{alg:PPI_QD} in our presentation. That is why we consider that both algorithms use the same underlying mechanism.
However, there are some differences between "QDPM" and "PPI" that may lead to different executions over the same game:
\begin{itemize}
	\item "QDPM" does not apply potential updates modifying the game. Instead, it carries the information in a potential $\mu$ (progress measure in the terminology of~\cite{BDM24}), which is updated in each iteration. The information carried by the potential $\mu$ is used in the other iterations by the algorithm. (By iteself, this does not provoke differences in executions.)
	
	\item QDPM does not initialise the quasi dominion $F$ to $\compl{\negG{\game}}$. Instead, $F$ is the set of positions which are assigned value $>0$ by the potential $\mu$ coming from previous iterations ($F = \overline{\mu^{-1}(0)}$). 
	In order to enlarge this set, a first small update (corresponding to a potential update of $\PsiFirst$) is applied to $F$, this corresponds to $\mathsf{prg}_0(\overline{\mu^{-1}(0)})$ in~\cite[p.7]{BDM24}).
	
	\item It is important to notice that due to this first initialization step applying a first potential $\mathsf{prg}_0$, the games treated by "PPI" and "QDPM" slightly differ. Over several iterations, the behaviour of both algorithms may therefore diverge. We observe empirically that while there may be rare differences between the number of iterations of the two algorithms, they remain negligible compared to the total number, and are not biased towards one or the other algorithm.
	
	\item Importantly, "QDPM" includes some smart implementation techniques to avoid considering all vertices in the computation of $\mathsf{prg}_0$ and $\mathsf{prg}_+$~\cite[Sect.~5]{BDM24}, improving their theoretical complexity upper bound to match the one of~\cite{BCDGR11}.
\end{itemize}
 
	 \subsection{Implementation differences comparison between state-weighted and vertex-weighted games}\label{app:implementationDifferences}
	We now propose an explanation why QDPM, as implemented by~\cite{BDM24}, performing an order of magnitude quicker than our implementation of PPI, basing on the fact that the QDPM is based on vertex-weighted games whereas PPI is based on edge-weighted ones.
Recall that we study games where weights are exponential (this is also the case when translating from parity games with linearly many priorities); therefore essentially the whole runtime is spent on performing such operations, which are either additions or comparisons.

These are broken into three categories:
\begin{enumerate}[1.]
    \item Updating the total potential $\Phi$ of each vertex. This requires roughly $n$ additions per iteration.
    \item Insertions in priority queues. This requires roughly $n\log n$ comparisons.
    \item Weight comparisons. This is where the difference lies. Here, we should compare the modified weights of two outgoing edges $vv_1$ and $vv_2$ from a given vertex $v$. In the edge-weighted scenario, this amounts to comparing $-\Phi(v) + \Phi(v_1) + w(vv_1)$ with $-\Phi(v) + \Phi(v_2) + w(vv_2)$, or equivalently $\Phi(v_1) + w(vv_1)$ with $\Phi(v_2) + w(vv_2)$. This requires \emph{2 additions and 1 comparison}, which amounts overall to roughly $2m$ additions and $m$ comparisons per iteration.
    In contrast, in the vertex-weighted scenario, $w(vv_1)=w(vv_2)$, so it is enough to compare $\Phi(v_1)$ with $\Phi(v_2)$, leading to $m$ comparisons, which saves on $2m$ costly additions per iteration.
\end{enumerate}

In total, we get the following numbers:

\vskip1em
\begin{center}
    \begin{tabular*}{0.8\linewidth}{l | c | c}
        type$\qquad $ & $\qquad $comparisons$\qquad $ &$\qquad $ additions \\ \hline
        edge-weighted $\qquad $& $n \log n$ & $\qquad $$n + 2m$ \\
        vertex-weighted & $n \log n$ & $\qquad $$n$ \\
    \end{tabular*}
\end{center}

\vskip1em

To give concrete estimates, we have compared runtimes between additions and comparisons (in the GMP libraries), for weights corresponding to the our biggest instances, reporting a ratio of over 4 orders of magnitude ($10^4$).
For sparse games ($m=2n$) this explains a factor of roughly 5 between the two implementations, which is more that the difference in runtimes.
For dense games, this gives a linear factor in $n$ on the number of additions (although arguably dense games typically have smaller weights).
\section{Friedmann's family of examples}\label{app:friedmann}

We include the performance (on number of iterations) of our algorithms against the family of examples proposed by Friedmann~\cite{Friedmann11b} (Figure~\ref{fig:Friedmann}). We also ran the same experiments in the randomized setting, where instances are first perturbed by random potentials sampled according to a normal distribution (Figure~\ref{fig:perturbedFriedmann}).

\begin{figure}[h]
	\begin{center}
		\includegraphics[width=0.65\linewidth]{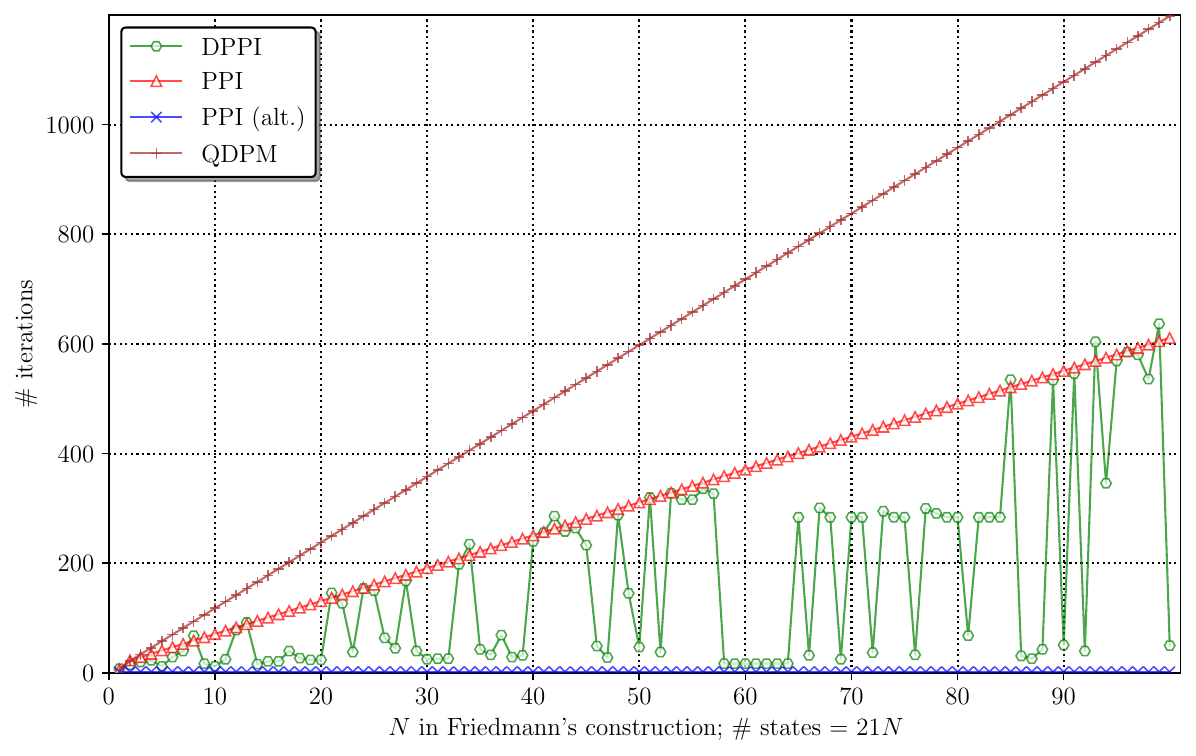}
		\caption{Number of iterations against Friedmann's hard examples.}
		\label{fig:Friedmann}
	\end{center}
\end{figure}

It is not surprising that "PPI-alt" performs a constant (namely, 2) number of iterations, because the dual algorithm immediately attracts the whole game to the single negative cycle (the instances are by no means designed to be resilient to such dual algorithms).
We indeed observe that QDPM and PPI perform a linear number of iterations, as claimed in the conclusion.
Remarkably, DPPI performs a constant number of iterations on roughly half of the instances, while on a few instances it performs slighly more iterations than PPI (this is not a contradition to Theorem~\ref{thmt@@lemComparingPotentials} as explained just above it).

\begin{figure}[h]
	\begin{center}
		\includegraphics[width=0.65\linewidth]{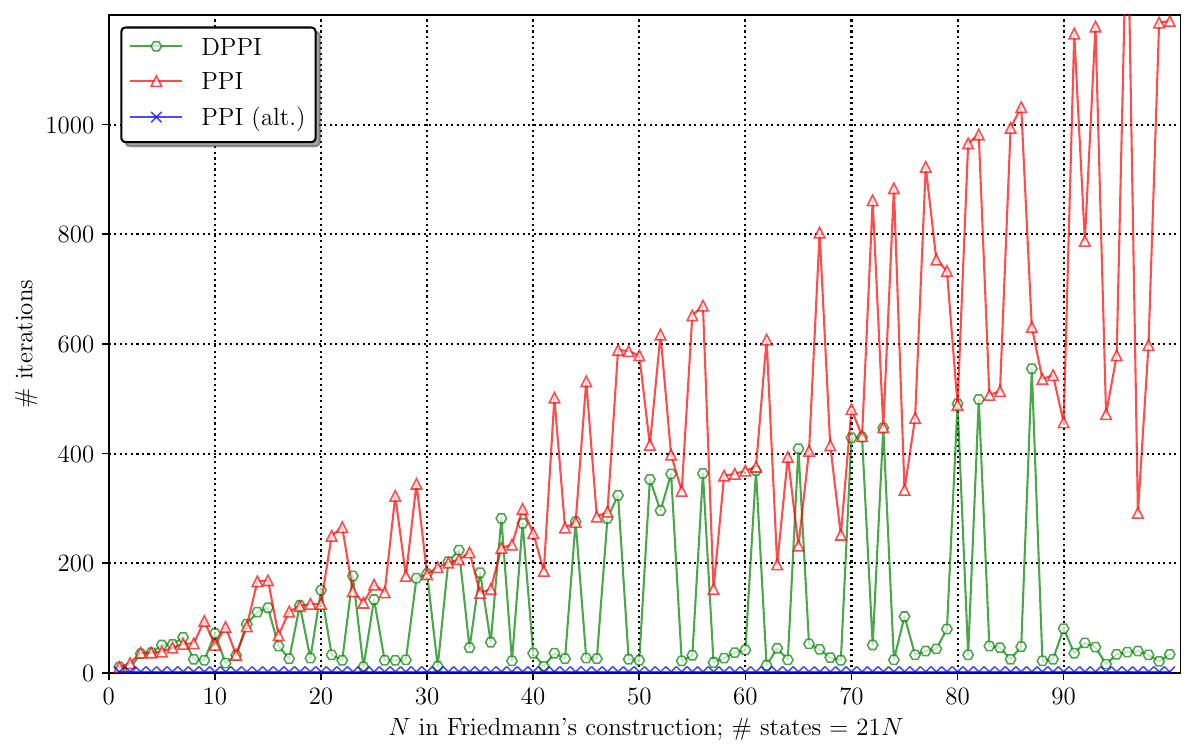}
		\caption{Number of iterations against Friedmann's hard examples in the randomized setting with initial perturbations.}
		\label{fig:perturbedFriedmann}
	\end{center}	
\end{figure}

In the randomized setting, we observe some speedup (for PPI) on some of the instances, which shows that it could make sense to run (in parallel) the algorithm on perturbated inputs.
However, we remark that in most of the cases the number of iterations is noticeably increased.
This constitutes by no means a serious experimental study of this phenomenon, which we leave to future work.
In particular, it would be more meaningful to run this experiment on instances requiring super linear number of updates (which are not available at the moment).

\else\fi

% \section{Examples of games and executions}\label{app:gamesExamples}
% \input{appendix-gameExamples}
\end{document}